\documentclass[10pt]{article}
\usepackage{amsmath, amsthm}
\usepackage[doublespacing]{setspace}
\usepackage{geometry}
\usepackage{amsfonts}
\usepackage{sectsty}
\usepackage{booktabs}
\usepackage{fancyhdr}
\usepackage{natbib}
\usepackage{url}
\usepackage[pdftex]{graphicx}
\usepackage[pdftex]{color}
\usepackage[colorlinks=true, linkcolor=blue, citecolor=blue, urlcolor=blue]{hyperref}
\usepackage{multirow}

\usepackage[boxruled,linesnumbered]{algorithm2e}

\newcommand{\opO}{\operatorname{O}}
\newcommand{\GL}{\operatorname{GL}}
\newcommand{\optr}{\operatorname{tr}}

\newcommand{\opHaar}{\operatorname{Haar}}

\newcommand{\opdet}{\operatorname{det}}

\newcommand{\opdiag}{\operatorname{diag}}
\newcommand{\ophBayes}{\operatorname{hBayes}}

\DeclareMathOperator*{\argmin}{arg\,min}

\renewcommand{\vec}[1]{\boldsymbol{\mathrm{#1}}}

\newtheorem{thm}{Theorem}
\newtheorem{lemma}[thm]{Lemma}
\newtheorem{proposition}[thm]{Proposition}
\newtheorem{corollary}[thm]{Corollary}
\theoremstyle{remark}
\newtheorem*{remark}{Remark}

\theoremstyle{definition}
\newtheorem{definition}[thm]{Definition}

\begin{document}
\title{Hierarchical Bayesian Estimation of Covariance Matrices}

\author{Daniel Yekutieli, Daniel Xiang, Jonas Wallin, Malgorzata Bogdan}

\baselineskip=18pt

\maketitle

\begin{abstract}
We develop a hierarchical Bayesian framework for covariance matrix estimation built on a key observation:
while equivariance under the full general linear group GL($p$) is well known, it is an extremely restrictive
property. Estimators equivariant to GL($p$) are limited to scalar multiples of the sample covariance matrix
and carry considerably larger risks than shrinkage estimators. By contrast, commonly used shrinkage
estimators, including the Haff empirical Bayes estimator and the
Ledoit--Wolf estimators, are all equivariant under the smaller orthogonal group O($p$). Exploiting this
structure, we establish that the Haar measure Bayes rule in an oracle eigenvalue model is the minimum risk
estimator within the class of O($p$)-equivariant estimators, and derive oracle Bayes rules for the covariance
and precision matrices under the squared Frobenius, Stein, and squared Stein loss functions. These oracle
rules serve as theoretical benchmarks for all O($p$)-equivariant estimators.
To approximate them when the true eigenvalues are unknown, we introduce a hierarchical
Bayes model that places a finite P\'olya tree prior on the eigenvalue distribution and uses Gibbs sampling
to generate posterior draws, yielding both shrinkage estimates for the eigenvalues and approximations to the oracle Bayes rules.
In simulations, the resulting estimators come close to the risk of the oracle rules and have smaller risk
than commonly used estimators, particularly when $p/n$ is large.
In an application to monthly returns of 49 US industry portfolios, the proposed estimators had the highest
out-of-sample likelihood of all the estimators we compared, and their minimum-variance portfolios had
variance comparable to that of estimators designed specifically for portfolio selection.
\end{abstract}

\section{Introduction}

This paper presents a hierarchical Bayes (hBayes) approach for estimating covariance matrices of
Multivariate Normal (MVN) distributions. The problem has a rich history. \citet{JamesStein1961} established
equivariance of covariance matrix estimation under multiplication by nonsingular matrices. \citet{Stein1975}
proposed estimates that incorporate robust shrinkage of the eigenvalue vector, demonstrating considerably
smaller risk in simulations. \citet{Haff1980} introduced empirical Bayes estimators expressed as weighted
combinations of the sample covariance matrix and an arbitrary positive definite matrix, and \citet{Haff1991}
studied the class of estimators invariant under orthogonal transformations, deriving the variational form of Bayes rules
that minimize risk for standard matrix loss functions. More recently, \citet{LedoitWolf2004} proposed a
linear shrinkage estimator with analytically optimal weights, \citet{LedoitWolf2012} extended this to
nonlinear shrinkage of the sample eigenvalues, and \citet{LedoitWolf2020} gave an analytical version of
the nonlinear shrinkage estimator.
The idea of employing a finite P\'olya tree prior to nonparametrically learn a shrinkage rule, which is a
key feature of our hBayes framework, was previously explored in \citet{WeinsteinEtAl2025}. In that work, the
authors place a finite P\'olya tree prior on the coefficients of a generalized linear model to build a
hierarchical Bayes estimator that approximates an oracle shrinkage rule characterized by permutation
invariance of the coefficient vector.

Our approach to covariance matrix estimation under the Normal model is built on two
generative models defined via eigendecomposition. The first is an oracle model in which the eigenvalue
vector is treated as known and fixed, and the eigenvector matrix is drawn uniformly from the orthogonal
group. The second is a hBayes model that jointly samples the eigenvector matrix and
places a finite P\'olya tree prior on the eigenvalues. Exploiting equivariance properties of standard loss
functions and shrinkage estimators under orthogonal transformations, we establish the oracle model's Bayes
rules as theoretical benchmarks that minimize risk within the class of orthogonally equivariant estimators
for both covariance and precision matrix estimation. The hBayes model generates posterior samples of the
eigenvalue vector and eigenvector matrix, which we use to derive shrinkage estimates for the sorted
eigenvalues and to approximate the oracle Bayes rules when the eigenvalues are unknown.

The remainder of the paper is organized as follows. Section~\ref{sec:equivariance-to-O(p)} describes the
setting and equivariance under the orthogonal group, and Appendix~\ref{app:equivariance} reviews the general
results on equivariance that we use. Section~\ref{sec:oracle-eigenvalues-bayes-rules} presents the oracle
model and its Bayes rules for the covariance matrix, and Section~\ref{sec:oracle-implementation} describes its
implementation. The corresponding results for the precision matrix are given in Appendix~\ref{app:precision}.
Section~\ref{sec:hbayes} develops the hBayes approach. Section~\ref{sec:simulations} presents simulation
results and Section~\ref{sec:data} applies the methods to monthly returns of industry portfolios. 
We conclude with a discussion in Section~\ref{sec:discussion}. Results
well established in the statistical literature are stated as Theorems; proofs are included in the text only
when instructive, with more technical derivations deferred to Appendix~\ref{sec:proofs}.
Appendix~\ref{app:spike} describes an exact posterior sampler for the single-spike model, which
Appendix~\ref{app:cr} uses to check the oracle sampler while constructing credible regions for the
leading eigenvector. Appendix~\ref{app:risk-full} gives the full results of the risk comparison.

\section{Setting and equivariance}
\label{sec:equivariance-to-O(p)}

\subsection{Model and loss functions}
\label{subsec:equivariance-GL}
We assume the observed data is  $\vec{\bf X}_{n \times p}$ consisting of  i.i.d.
row vectors, ${\bf X}^i \sim N_p( \vec{0}, \vec{\Sigma})$ with positive definite $\vec{\Sigma}_{p \times p}$.
We use the eigendecomposition to specify the covariance matrix $\vec{\Sigma}  = \vec{\Gamma} \vec{\Lambda} \vec{\Gamma}^\top$.
Here, $\vec{\Lambda}$ is the diagonal matrix whose diagonal is the vector of sorted eigenvalues,
$\vec{\lambda} = ( \lambda_1 \ge \cdots \ge  \lambda_p > 0)$,
and  $\vec{\Gamma}$ is the corresponding orthogonal eigenvector matrix.
We express the density of $\vec{X} = \vec{x}$,
\begin{eqnarray}
 f  (   \vec{x} | \vec{\Sigma} )   & =  & \Pi_{i = 1}^n  (2 \pi)^{-p/2} \opdet( \vec{\Sigma})^{-1/2} \exp \{ -\vec{x}^i \vec{\Sigma}^{-1} (\vec{x}^i)^\top / 2\} \nonumber \\
  & =  & (2 \pi)^{-np/2} \opdet( \vec{\Sigma})^{-n/2} \exp \{ \optr( -\vec{x} \vec{\Sigma}^{-1} \vec{x}^\top / 2)\} \nonumber \\
  & =  & (2 \pi)^{-np/2}  \left(  \Pi_{i = 1}^p \lambda_i \right)^{-n/2} \exp \{ \optr( -  \vec{\Lambda}^{-1} \vec{\Gamma}^\top  ( n \vec{S} ) \vec{\Gamma} / 2)\} \label{norm-lik}
\end{eqnarray}
where  $\vec{S}  =  \vec{X}^\top \vec{X}  / n$ is the sample covariance.

\begin{definition}\label{def:equivariance}
Let ${\cal G}$ be a group of transformations, in which each $g \in {\cal G}$ specifies a
transformation on the data, $\vec{X} \mapsto g(\vec{X})$, and on the parameter space,
$\theta \mapsto \bar g(\theta)$.
(a) A family of densities $f(\vec{x} \mid \theta)$, $\theta \in \Theta$, is equivariant under
${\cal G}$ if, for all $g \in {\cal G}$, $\vec{X} \sim f(\vec{x} \mid \theta)$ implies
$g(\vec{X}) \sim f(\vec{x} \mid \bar g(\theta))$.
(b) An estimator $\hat\theta$ is equivariant under ${\cal G}$ if
$\hat\theta(g(\vec{X})) = \bar g(\hat\theta(\vec{X}))$ for all $g \in {\cal G}$.
(c) A loss function $L(\hat\theta, \theta)$ is invariant under
${\cal G}$ if $L(\bar g(\hat\theta), \bar g(\theta)) = L(\hat\theta, \theta)$ for all $g \in {\cal G}$.
\end{definition}

For covariance estimation, $\theta = \vec{\Sigma}$ and $\Theta$ is the set of $p \times p$ positive definite
matrices. We consider equivariance to $\GL(p)$, the
general linear group of $p \times p$ nonsingular matrices. Each nonsingular matrix $\vec{A}$ is
associated with the transformation $g_{\vec{A}}(\vec{X}) = \vec{X}\vec{A}$ on the data space and
the transformation $\bar{g}_{\vec{A}}(\vec{\Sigma}) = \vec{A}^\top \vec{\Sigma} \vec{A}$ on the
parameter space. The map $\bar g_{\vec{A}}$ is applied in the same way to positive semidefinite matrices,
the values of the estimators. We will focus on equivariance to $\opO(p)$, the subgroup of orthogonal matrices.

\medskip
We consider three commonly used loss functions: the squared Frobenius loss
$L_0(\hat{\boldsymbol\Sigma}, \boldsymbol\Sigma) = \|\hat{\boldsymbol\Sigma} - \boldsymbol\Sigma\|_F^2$,
the Stein loss
$L_1(\hat{\boldsymbol\Sigma}, \boldsymbol\Sigma) = \operatorname{tr}(\hat{\boldsymbol\Sigma}\boldsymbol\Sigma^{-1})
  - \log\det(\hat{\boldsymbol\Sigma}\boldsymbol\Sigma^{-1}) - p$,
and the squared Stein loss
$L_2(\hat{\boldsymbol\Sigma}, \boldsymbol\Sigma) = \operatorname{tr}\{(\hat{\boldsymbol\Sigma}\boldsymbol\Sigma^{-1} - \mathbf I)^2\}$.
The following results are classical (see \citet{JamesStein1961} and \citet{Eaton1989}).

\begin{thm} \label{lemm1a}
\begin{enumerate}
\item[(a)] The $N_p( \vec{0}, \vec{\Sigma})$ distribution is equivariant to $\GL(p)$.
\item[(b)] The sample covariance matrix $\vec{S}$ is equivariant to $\GL(p)$.
\item[(c)] The $L_0$ loss is invariant to $\opO(p)$.
\item[(d)] The $L_1$ loss is invariant to $\GL(p)$.
\item[(e)] The $L_2$ loss is invariant to $\GL(p)$.
\end{enumerate}
\end{thm}

\begin{thm} \label{prop122}
Estimators that are equivariant to $\GL(p)$ are of the form $ \hat{\vec{\Sigma}} (\vec{S}) = c \vec{S}$ for $c > 0$.
 \end{thm}

Theorem~\ref{prop122} shows that equivariance under $\GL(p)$ leaves only scalar multiples of $\vec{S}$,
so we work with the smaller group $\opO(p)$. Two general results on equivariance, stated and proved
in Appendix~\ref{app:equivariance}, are used throughout. If the family of densities and an estimator are
equivariant and the loss is invariant under a group, then the risk of the estimator is constant on each
orbit of the group in the parameter space (Theorem~\ref{cor1}). If the group is also compact, consider the
prior obtained by applying a group element drawn from the Haar measure to a fixed parameter value. The Bayes
rule under this prior is then equivariant, and among equivariant estimators it has the smallest risk at
every parameter value in the orbit (Theorem~\ref{cor11}).

\subsection{Equivariance under the orthogonal group}
\label{subsec:equivariance-O}

Expressing the covariance matrix $\vec{\Sigma} =   \bar{g}_{\vec{\Gamma}^\top} ( \vec{\Lambda})$
 reveals that $\opO(p)$ is transitive on the set of all covariance matrices with eigenvalue vector $\vec{\lambda}$.
 We denote the  eigendecomposition of the sample covariance matrix $\vec{S}   = \vec{V} \vec{L} \vec{V}^\top$.
Here,  $\vec{L}$ is the diagonal matrix of ordered sample eigenvalues,
$\vec{l} = ( l_1 \ge  \cdots \ge  l_p)$,
and $\vec{V} = ( \vec{V}_1, \cdots, \vec{V}_p)$ is the corresponding orthogonal sample eigenvector matrix.

We consider estimators $\hat{\vec{\Sigma}}(\vec{S})$, and likewise $\hat{\vec{\Omega}}(\vec{S})$, taking positive semidefinite values, that depend on
the data through the sufficient statistic $\vec{S}$. Equivariance under $\opO(p)$ then means
$\hat{\vec{\Sigma}}(\vec{R}^\top\vec{S}\vec{R}) = \vec{R}^\top\hat{\vec{\Sigma}}(\vec{S})\vec{R}$
for all $\vec{R} \in \opO(p)$.
For $n < p$, $\vec{S}$ is singular, with $l_i = 0$ for $i > n$, and $\vec{V}_{n+1}, \dots, \vec{V}_p$ may be any
orthonormal basis of its null space. For $n \le p-2$ this basis is not unique, which forces the
corresponding diagonal entries of $\hat{\vec{\Sigma}}(\vec{L})$ to be equal.

\begin{proposition} \label{prop12}
If a covariance matrix estimator $\hat{\vec{\Sigma}}(\vec{S})$ is equivariant to $\opO(p)$
then $\hat{\vec{\Sigma}} ( \vec{S})  = \vec{V}  \hat{\vec{\Sigma}} (\vec{L} )  \vec{V}^\top$
with  $\hat{\vec{\Sigma}} (\vec{L}) = \opdiag ( d_1, \cdots,  d_p)$.
If  $n \le p-2$ then $d_{n+1} = \cdots = d_p$.
 \end{proposition}

Note that the $N_p ( \vec{0}, \vec{\Sigma})$ density may also be specified by the
precision matrix $\vec{\Omega} = \vec{\Sigma}^{-1}$, with eigendecomposition $\vec{\Omega} = \vec{\Gamma} \vec{\Lambda}^{-1} \vec{\Gamma}^\top$
for $\vec{\Lambda}^{-1} = \opdiag( 1/\lambda_1 \le  \cdots \le 1/\lambda_p)$.
Lemma \ref{lemm2a} establishes equivariance to $\opO(p)$
of applying the transformation $\bar{g}_{\vec{R}}$ to the precision matrix, and equivariance to $\opO(p)$
of the sample precision matrix $\vec{S}^{-1} ( \vec{X})$.

\begin{lemma} \label{lemm2a}
For $\vec{R} \in \opO(p)$,
(a) $\bar{g}_{\vec{R} }( \vec{\Sigma}^{-1}) = (\bar{g}_{\vec{R} }( \vec{\Sigma}))^{-1}$,
and (b) for $n \ge p$,  $\vec{S}^{-1} ( g_{\vec{R}} ( \vec{X})) = \bar{g}_{\vec{R}} ( \vec{S}^{-1} ( \vec{X}))$.
\end{lemma}

\begin{corollary} \label{cor44}
If a precision matrix estimator $\hat{\vec{\Omega}}(\vec{S})$ is equivariant to $\opO(p)$ then
$\hat{\vec{\Omega}} ( \vec{S})  = \vec{V}  \hat{\vec{\Omega}} (\vec{L} )  \vec{V}^\top$
with  $\hat{\vec{\Omega}} (\vec{L}) = \opdiag ( d_1, \cdots,  d_p)$.
If  $n \le p-2$ then $d_{n+1} = \cdots = d_p$.
\end{corollary}

\begin{proof}
We may extend Proposition \ref{prop12} to precision matrix estimators because the only condition needed from the estimator in
its proof was equivariance under $\opO(p)$.
\end{proof}

\begin{lemma}\label{lem:diag-equiv}
    Let $\vec{S} = \vec{V} \vec{L} \vec{V}^\top$ be the eigendecomposition of the sample covariance matrix. Any estimator of the form $\hat{\vec{\Sigma}}(\vec{S}) = \vec{V} \hat{\vec{\Lambda}}(\vec{L}) \vec{V}^\top$, where $\hat{\vec{\Lambda}}(\vec{L}) = \opdiag(\hat{\lambda}_1(\vec{L}),\dots,\hat{\lambda}_p(\vec{L}))$ is a diagonal matrix, is orthogonally equivariant.
\end{lemma}

\begin{proof}
Let $\vec{R} \in \opO(p)$ be an arbitrary orthogonal matrix. The eigendecomposition of 
$\vec{R}\vec{S}\vec{R}^\top$ is $(\vec{R}\vec{V})\vec{L}(\vec{R}\vec{V})^\top$, which implies 
$\hat{\vec{\Sigma}}(\vec{R}\vec{S}\vec{R}^\top) = (\vec{R}\vec{V})\hat{\vec{\Lambda}}(\vec{L})(\vec{R}\vec{V})^\top = 
\vec{R} (\vec{V} \hat{\vec{\Lambda}}(\vec{L})\vec{V}^\top)\vec{R}^\top = 
\vec{R} \hat{\vec{\Sigma}}(\vec{S}) \vec{R}^\top$, that is, $\hat{\vec{\Sigma}}(\bar g_{\vec{R}^\top}(\vec{S})) 
= \bar g_{\vec{R}^\top}(\hat{\vec{\Sigma}}(\vec{S}))$.
\end{proof}

\begin{corollary}
\label{cor:common-invariant-estimators}
    The following covariance estimators are orthogonally equivariant:
    \begin{enumerate}
        \item MLE, where $\hat{\lambda}_j(\vec{L}) = l_j$.

\item \citet{Stein1975} eigenvalue shrinkage estimator
  $\hat\lambda_j(\vec{L}) = n l_j / (n + p + 1 - 2j)$, and adaptive eigenvalue
  shrinkage estimator
  $\hat\lambda_j(\vec{L}) = n l_j \big/ \bigl(n - p + 1 + 2 l_j \sum_{i \neq j} (l_j - l_i)^{-1}\bigr)$,
  together with its isotonized version \citep{LinPerlman1985}, which adjusts
  the $\hat\lambda_j(\vec{L})$ so that they are positive and in the same order as the
  $l_j$. Since the adjustment depends on $\vec{L}$ only, the isotonized estimator is
  also of the form in Lemma~\ref{lem:diag-equiv}.

        \item \citet{Haff1980} empirical Bayes estimator
        $\hat\lambda_j(\vec{L}) = \frac{n}{n+p+1}\Big(l_j + \frac{p-1}{(n-p+3)\sum_{i=1}^p l_i^{-1}}\Big)$
        \item \citet{LedoitWolf2004} linear shrinkage estimator $\hat{\lambda}_j(\vec{L}) = \hat{w}\hat{c} + (1-\hat{w})l_j$, where $\hat{w},\hat{c}$ are functions of $\vec{L}$.

\item \citet{LedoitWolf2012} and \citet{LedoitWolf2020} nonlinear shrinkage
  estimators, $\hat\lambda_j(\vec{L}) = \tilde d_{p,n}(l_j; F_p)$, where
  $\tilde d_{p,n}(\cdot\,; F_p)$ is a nonlinear function parameterized by the
  empirical distribution of the sample eigenvalues
  $F_p = p^{-1}\sum_{i=1}^p \delta_{l_i}$ and the ratio $p/n$; in the 2012
  estimator it is computed from the population eigenvalues estimated by QuEST
  \citep{LedoitWolf2015}, and in the 2020 estimator analytically.

\item \cite*{donoho2018optimal} estimator $\hat{\lambda}_j(\vec{L})=\eta^*(l_j, \gamma)$, where $\eta^*$ is a non-linear shrinkage function depending on the loss function and the limiting aspect ratio $p/n \to \gamma \in (0,1]$.
 \end{enumerate}
\end{corollary}

On the other hand, several commonly used estimators are not orthogonally
equivariant. The minimax estimator of \citet{JamesStein1961} is built from the
Cholesky factor of $\vec{S}$ and is equivariant under lower-triangular
transformations, but not under $\opO(p)$. Haff's empirical Bayes estimator with
non-isotropic prior \citep{Haff1980}, the graphical lasso precision matrix
estimator \citep{friedman2008sparse} and the banded sample covariance estimator
\citep{bickel2008regularized} use coordinate-level structure of $\vec{\Sigma}$ or $\vec{\Omega}$,
such as sparsity or bandedness in a particular basis. Estimators used for asset
returns are tied to the coordinates in a similar way: shrinkage toward the
single-index model \citep{LedoitWolf2003}, whose market factor is an average of
the coordinates, or toward a constant-correlation model
\citep{LedoitWolf2004jpm}, and factor models with diagonal residual covariance
\citep{Chan1999, Fan2013}, all use the diagonal of $\vec{\Sigma}$ in the given basis.
None of these estimators is equivariant under orthogonal rotations.

\medskip
In the next section, we derive the Bayes estimator for $\vec{\Sigma}$ in an oracle model where the true eigenvalues of $\vec{\Sigma}$ are known and fixed and the eigenvector matrix $\vec{\Gamma}$ is drawn from the Haar measure on $\opO(p)$. This (oracle) decision rule minimizes the risk in the class of decision functions that are equivariant under $\opO(p)$; it is an oracle in the sense that it requires the true eigenvalue vector $\vec{\lambda}$. Section \ref{sec:hbayes} discusses methods for approximating the oracle Bayes rule using estimates for $\vec{\lambda}$ based on hierarchical Bayes modeling.

\section{Oracle eigenvalue model}
\label{sec:oracle-eigenvalues-bayes-rules}

The oracle  model is a generative model for the data that assumes that $\vec{X} \sim f ( \vec{x} | \vec{\Sigma})$
for $\vec{\Sigma} = \vec{\Gamma} \vec{\Lambda} \vec{\Gamma}^\top$
with  known fixed eigenvalue vector $\vec{\lambda}$ and eigenvector matrix $\vec{\Gamma}$ sampled from  $\pi^{\opHaar}$,
the Haar measure on $\opO(p)$.
The conditional distribution of $\vec{\Gamma}$ given $\vec{X}  = \vec{x}$ for the oracle model is
 \begin{equation}  \label{eq34}
\pi^{\opHaar} ( \vec{\Gamma} | \vec{x}, \vec{\lambda})
= \frac{ f  (   \vec{x} | \vec{\Gamma} , \vec{\lambda} ) \pi^{\opHaar} ( \vec{\Gamma} |  \vec{\lambda})}{ f^{\opHaar} (\vec{x} | \vec{\lambda})}
= \frac{ f  (   \vec{x} | \vec{\Sigma} ) \pi^{\opHaar} ( \vec{\Gamma} )}{ f^{\opHaar} (\vec{x} | \vec{\lambda})},
\end{equation}
where  $f^{\opHaar} (\vec{x} | \vec{\lambda})$ is the marginal density of $\vec{X}$ in the oracle model.
As noted in \citet{hoff2009simulation},  the use of the uniform  prior distribution implies that $\vec{\Gamma}$ has
a matrix Bingham posterior distribution in the oracle model,
 \begin{equation}  \label{eq34a}
\pi^{\opHaar} ( \vec{\Gamma} | \vec{x}, \vec{\lambda})  \propto \exp ( \optr( -  \vec{\Lambda}^{-1} \vec{\Gamma}^\top  ( n \vec{S} ) \vec{\Gamma} / 2)),
\end{equation}
which is proportional to the normal likelihood (\ref{norm-lik}).
According to Theorem 5 in \cite{jupp1979maximum},
the sample eigenvector matrix $\vec{V}$ is a mode of $\pi^{\opHaar} ( \vec{\Gamma} | \vec{x}, \vec{\lambda})$.

\medskip
The next two lemmas identify symmetries and therefore multiple modes
in $\pi^{\opHaar} ( \vec{\Gamma} | \vec{x}, \vec{\lambda})$.
Let ${\cal D}(p) \subset \opO(p)$ denote the diagonal matrices with entries $\pm 1$.
For $\vec{B} \in {\cal D}(p)$, $\vec{\Gamma}\vec{B}$ differs from $\vec{\Gamma}$ only in the signs of its columns.
Non-trivial symmetries arise
if the components of $\vec{\lambda}$ are non-distinct
and thus $\exists \vec{B} \in \opO(p) \setminus {\cal D}(p)$
with  $\vec{\Lambda} = \vec{B} \vec{\Lambda} \vec{B}^\top$,
or if $n \le p-2$, in which case $\vec{S}$ is singular (while $\vec{\Sigma}$ remains positive definite, so that
(\ref{norm-lik}) and (\ref{eq34a}) are unchanged), $l_i = 0$ for $i > n$, and thus
$\exists \vec{B} \in \opO(p) \setminus {\cal D}(p)$  with
$\vec{L} = \vec{B} \vec{L} \vec{B}^\top$.

\begin{lemma} \label{lemm4a}
For $\vec{B} \in \opO(p)$ such that $\vec{\Lambda} = \vec{B} \vec{\Lambda} \vec{B}^\top$,
$\pi^{\opHaar} (  \vec{\Gamma} \vec{B} | \vec{x}, \vec{\lambda}) =\pi^{\opHaar} ( \vec{\Gamma} | \vec{x}, \vec{\lambda})$.
\end{lemma}

\begin{lemma}  \label{lemm5a}
For $\vec{B} \in \opO(p)$ such that $\vec{L} = \vec{B} \vec{L} \vec{B}^\top$,
\[
\pi^{\opHaar} \left(  \vec{\Gamma}  ( \vec{\Gamma}^\top \vec{V} \vec{B}^\top \vec{V}^\top \vec{\Gamma} )  | \vec{x}, \vec{\lambda}\right)
=\pi^{\opHaar} ( \vec{\Gamma} | \vec{x}, \vec{\lambda}).
\]
\end{lemma}

\noindent
Setting $\vec{\Gamma}  =  \vec{V}$ in Lemma \ref{lemm5a} yields the following result.
\begin{corollary}
For $n \le p-2$, $\exists \vec{B} \in \opO(p) \setminus {\cal D}(p)$
with $\vec{L} = \vec{B} \vec{L} \vec{B}^\top$,
for which $\vec{V} \vec{B}^\top$ is also a  mode of
$\pi^{\opHaar} ( \vec{\Gamma} | \vec{x}, \vec{\lambda})$.
\end{corollary}

\begin{lemma} \label{cor34}
$\pi^{\opHaar} ( \vec{\Gamma} | \vec{x}, \vec{\lambda} )$ is equivariant to $\opO (p)$.
\end{lemma}

\begin{definition}   \label{oracle-lambda-BR-def}
The oracle Bayes rule is the estimator that minimizes the average risk for the  oracle eigenvalue model,
$E_{\vec{\Gamma} \sim \pi^{\opHaar}}  R ( L, \hat{\vec{\Sigma}}; \vec{\Sigma})$.
It is  derived by minimizing the posterior expected loss,
\begin{equation} \label{oracle-lambda-BR}
\hat{\vec{\Sigma}}^{\opHaar}   ( \vec{x}, \vec{\lambda})  := \argmin_{\hat{ \vec{\Sigma}} \succeq 0}  \left[
E_{\vec{\Gamma}  \sim  \pi^{\opHaar} ( \vec{\Gamma} | \vec{x}, \vec{\lambda} ) }
L ( \hat{\vec{\Sigma}}, \vec{\Sigma}) \right].\end{equation}

\end{definition}

\noindent
The minimum in (\ref{oracle-lambda-BR}) is taken over positive semidefinite matrices
(Section~\ref{subsec:equivariance-O}), with
$L_1(\hat{\vec{\Sigma}}, \vec{\Sigma}) = \infty$ for singular $\hat{\vec{\Sigma}}$. For each of the losses
$L_0$, $L_1$ and $L_2$, the posterior expected loss is strictly convex and tends to infinity as
$\|\hat{\vec{\Sigma}}\| \to \infty$ (and, under $L_1$, as $\hat{\vec{\Sigma}}$ approaches a singular matrix), 
so the oracle Bayes rule exists and is unique. The set of positive semidefinite matrices is invariant under $\opO(p)$, so Lemma~\ref{prop2w}
and Proposition~\ref{prop3} below apply with this set of estimator values.

\begin{lemma} \label{prop2w}
If the loss function is invariant to  $\opO (p)$ then the oracle Bayes rule is equivariant
to $\opO (p)$.
\end{lemma}

\begin{proof}
Lemma  \ref{cor34} implies that  for an invariant loss function the posterior expected loss in (\ref{oracle-lambda-BR}) is invariant under $\opO (p)$.
Therefore the oracle Bayes rule, its minimizer,  is equivariant under $\opO (p)$.
 \end{proof}

\begin{proposition} \label{prop3}
If the loss function is invariant to $\opO (p)$ then the oracle Bayes rule
is the equivariant estimator under $\opO(p)$, taking positive semidefinite values, with the smallest risk.
\end{proposition}

\begin{proof}
This result follows from Lemma \ref{prop2w} and Theorem \ref{cor11}.
\end{proof}

\subsection{Oracle Bayes rules}\label{subsec:oracle-rules}

Let $\mathbf S = \mathbf V\mathbf L\mathbf V^\top$ and write
$\tilde{\boldsymbol\Gamma} = \mathbf V^\top\boldsymbol\Gamma$. We express the
covariance and precision matrices in the sample eigenbasis,
\[
  \tilde{\boldsymbol\Sigma} = \mathbf V^\top\boldsymbol\Sigma\mathbf V
    = \tilde{\boldsymbol\Gamma}\boldsymbol\Lambda\tilde{\boldsymbol\Gamma}^\top,
  \qquad
  \tilde{\boldsymbol\Omega} = \mathbf V^\top\boldsymbol\Omega\mathbf V
    = \tilde{\boldsymbol\Gamma}\boldsymbol\Lambda^{-1}\tilde{\boldsymbol\Gamma}^\top,
\]
and denote posterior expectations under $\pi^{\mathrm{Haar}}(\boldsymbol\Gamma \mid \mathbf x, \boldsymbol\lambda)$
by $\mathrm E[\,\cdot \mid \mathbf x]$. For a vector $\mathbf d$, $\operatorname{diag}(\mathbf d)$ denotes the
diagonal matrix with diagonal $\mathbf d$, and for a square matrix $\mathbf M$, $\operatorname{diag}(\mathbf M)$
denotes the vector of its diagonal entries. By Lemma~\ref{prop2w} and Proposition~\ref{prop12} the oracle
Bayes rules have the form $\hat{\boldsymbol\Sigma} = \mathbf V\operatorname{diag}(\mathbf d)\mathbf V^\top$,
with $d_{n+1} = \cdots = d_p$ when $n < p$. To treat both cases at once, let
$m = n + 1$ if $n < p$ and $m = p$ otherwise, and let
$\mathbf G \in \{0,1\}^{p \times m}$ be the matrix that maps the free entries to
$\mathbf d$: for $n \ge p$, $\mathbf G = \mathbf I_p$; for $n < p$, $G_{kk} = 1$ for
$k \le n$ and $G_{k,n+1} = 1$ for $k > n$, all other entries zero. Thus
$\mathbf d = \mathbf G\boldsymbol\delta$ with $\boldsymbol\delta \in \mathbb R^m$, and
$\mathbf N = \mathbf G^\top\mathbf G = \operatorname{diag}(1, \dots, 1, p - n)$ when $n < p$. Proposition \ref{prop:oracle-bayes-rules} below summarizes the oracle Bayes rules for the three losses defined in Section~\ref{subsec:equivariance-GL}, and is proved in Appendix~\ref{sec:proofs}.

\begin{proposition}\label{prop:oracle-bayes-rules}
The oracle Bayes rule (\ref{oracle-lambda-BR}) for estimating $\boldsymbol\Sigma$ is
$\hat{\boldsymbol\Sigma} = \mathbf V\operatorname{diag}(\mathbf G\boldsymbol\delta)\mathbf V^\top$, where
\begin{enumerate}
  \item[(a)] under $L_0$:
    $\boldsymbol\delta = \mathbf N^{-1}\mathbf G^\top \operatorname{diag}\big(\mathrm E[\tilde{\boldsymbol\Sigma} \mid \mathbf x]\big)$,
    i.e.\ $d_k = \mathrm E[\tilde\Sigma_{kk} \mid \mathbf x]$ for $k \le \min(n,p)$;
  \item[(b)] under $L_1$:
    $\delta_i = 1 \big/ \big(\mathbf N^{-1}\mathbf G^\top \operatorname{diag}(\mathrm E[\tilde{\boldsymbol\Omega} \mid \mathbf x])\big)_i$,
    i.e.\ $d_k = 1/\mathrm E[\tilde\Omega_{kk} \mid \mathbf x]$ for $k \le \min(n,p)$;
  \item[(c)] under $L_2$: if the solution of the following system has positive entries, $\boldsymbol\delta$
    is that solution; otherwise $\boldsymbol\delta = \boldsymbol\delta^+$, the minimizer over
    $\boldsymbol\delta \ge \vec{0}$ (see below). The system is
    \[
      \mathbf G^\top \mathbf A\,\mathbf G\,\boldsymbol\delta = \mathbf G^\top \mathbf b,
      \qquad
      \mathbf A = \mathrm E[\tilde{\boldsymbol\Omega} \circ \tilde{\boldsymbol\Omega} \mid \mathbf x],
      \quad
      \mathbf b = \operatorname{diag}\big(\mathrm E[\tilde{\boldsymbol\Omega} \mid \mathbf x]\big),
    \]
    where $\circ$ denotes the elementwise product.
\end{enumerate}
For $n < p$ the common value of the last $p-n$ entries is, under $L_0$, the average
of $\mathrm E[\tilde\Sigma_{kk} \mid \mathbf x]$ over $k > n$ and, under $L_1$, the
reciprocal of the average of $\mathrm E[\tilde\Omega_{kk} \mid \mathbf x]$ over $k > n$.
\end{proposition}

\noindent
Under $L_0$ and $L_1$ the entries of $\boldsymbol\delta$ are positive, so $\hat{\vec{\Sigma}}$ is positive definite.
Under $L_2$, the solution of the linear system in (c) minimizes the posterior expected loss over all
symmetric matrices, and its entries are not guaranteed to be positive, although they were positive in all our
simulations. If some entry is not positive, the oracle Bayes rule is
$\hat{\vec{\Sigma}} = \vec{V}\operatorname{diag}(\mathbf G\boldsymbol\delta^+)\vec{V}^\top$, where
$\boldsymbol\delta^+$ minimizes the posterior expected loss over $\boldsymbol\delta \ge \vec{0}$. This rule is
singular, and it is characterized in Appendix~\ref{sec:proofs}. The same holds for the oracle Bayes
rules for $\vec{\Omega}$ in Corollary~\ref{cor:precision-rules}.

In terms of the eigenvector matrix, $\mathrm E[\tilde\Sigma_{kk} \mid \mathbf x]
= \sum_{j} \lambda_j \mathrm E[\tilde\Gamma_{kj}^2 \mid \mathbf x]$,
$\mathrm E[\tilde\Omega_{kk} \mid \mathbf x] = \sum_{j} \lambda_j^{-1}\mathrm E[\tilde\Gamma_{kj}^2 \mid \mathbf x]$ and
$\mathrm E[(\tilde\Omega_{k\ell})^2 \mid \mathbf x]
= \sum_{i,j}\lambda_i^{-1}\lambda_j^{-1}\mathrm E[\tilde\Gamma_{ki}\tilde\Gamma_{kj}\tilde\Gamma_{\ell i}\tilde\Gamma_{\ell j} \mid \mathbf x]$,
so the rules depend on the data through $\boldsymbol\lambda$ and the posterior moments of $\tilde{\boldsymbol\Gamma}$.
The rules are posterior counterparts of the finite-sample optimal estimators
of \citet{LedoitWolf2018}. Here, the $L_0$ rule estimates each eigenvalue by the
posterior mean of $\mathbf V_k^\top\boldsymbol\Sigma\mathbf V_k$, and the $L_1$
rule by the reciprocal of the posterior mean of
$\mathbf V_k^\top\boldsymbol\Sigma^{-1}\mathbf V_k$, whereas the finite-sample
optimal estimators plug in the unknown $\boldsymbol\Sigma$ directly.
In practice, the expectations are replaced by averages over posterior draws
$\tilde{\boldsymbol\Sigma}^{(g)}$ and $\tilde{\boldsymbol\Omega}^{(g)}$.

\begin{corollary}\label{cor:precision-rules}
The oracle Bayes rule for estimating $\boldsymbol\Omega$ under
$L_0(\hat{\boldsymbol\Omega}, \boldsymbol\Omega)$,
$L_1(\hat{\boldsymbol\Omega}, \boldsymbol\Omega)$ and
$L_2(\hat{\boldsymbol\Omega}, \boldsymbol\Omega)$ is
$\hat{\boldsymbol\Omega} = \mathbf V\operatorname{diag}(\boldsymbol\omega)\mathbf V^\top$ with
$\boldsymbol\omega = \mathbf G\boldsymbol\delta$,
where $\boldsymbol\delta$ is given by Proposition~\ref{prop:oracle-bayes-rules} with the roles of
$\tilde{\boldsymbol\Sigma}$ and $\tilde{\boldsymbol\Omega}$ interchanged:
for $k \le \min(n,p)$, $\omega_k = \mathrm E[\tilde\Omega_{kk} \mid \mathbf x]$ under $L_0$ and
$\omega_k = 1/\mathrm E[\tilde\Sigma_{kk} \mid \mathbf x]$ under $L_1$, and
$\mathbf G^\top\mathrm E[\tilde{\boldsymbol\Sigma}\circ\tilde{\boldsymbol\Sigma} \mid \mathbf x]\,\mathbf G\boldsymbol\delta
 = \mathbf G^\top\operatorname{diag}(\mathrm E[\tilde{\boldsymbol\Sigma} \mid \mathbf x])$ under $L_2$,
provided the solution has positive entries, and otherwise as in the paragraph after
Proposition~\ref{prop:oracle-bayes-rules}.
It has the minimum risk among orthogonally equivariant estimators of $\boldsymbol\Omega$ taking positive
semidefinite values.
\end{corollary}

\begin{remark}[Bayes rules under the hBayes posterior]
The derivation of Proposition~\ref{prop:oracle-bayes-rules} uses only that the posterior of
$\boldsymbol\Sigma$ is $\opO(p)$-equivariant, so the rules are diagonal in the basis
$\mathbf V$. Since the hBayes prior on $(\boldsymbol\lambda, \boldsymbol\Gamma)$ is
$\opO(p)$-invariant (Section~\ref{sec:hbayes-rules}), the same formulas give the Bayes rules under the hBayes posterior,
with $\tilde{\boldsymbol\Sigma}^{(g)} = \tilde{\boldsymbol\Gamma}^{(g)}\boldsymbol\Lambda^{(g)}\tilde{\boldsymbol\Gamma}^{(g)\top}$
and $\tilde{\boldsymbol\Omega}^{(g)} = \tilde{\boldsymbol\Gamma}^{(g)}(\boldsymbol\Lambda^{(g)})^{-1}\tilde{\boldsymbol\Gamma}^{(g)\top}$
computed from joint draws $(\boldsymbol\lambda^{(g)}, \boldsymbol\Gamma^{(g)})$.
\end{remark}

\begin{remark}[No data case]
When $n = 0$, the oracle Bayes rules for both the covariance and precision matrices reduce to a diagonal matrix with a common constant.
The specific value of this constant is determined by the true eigenvalues and the chosen loss function. The table below reports these constants for
the oracle covariance matrix estimators as well as for the oracle precision matrix estimators.\end{remark}

\begin{center}
\renewcommand{\arraystretch}{1.4}
\begin{tabular}{lcc}
\hline
Loss & $\hat{\boldsymbol\Sigma}$ diagonal & $\hat{\boldsymbol\Omega}$ diagonal \\
\hline
$L_0$ & $\textstyle\sum_j \lambda_j \,/\, p$
          & $\textstyle\sum_j \lambda_j^{-1} \,/\, p$ \\
$L_1$     & $p \,/\, \textstyle\sum_j \lambda_j^{-1}$
          & $p \,/\, \textstyle\sum_j \lambda_j$ \\
$L_2$ & $\textstyle\sum_j \lambda_j^{-1} \,/\, \sum_j \lambda_j^{-2}$
          & $\textstyle\sum_j \lambda_j \,/\, \sum_j \lambda_j^{2}$ \\
\hline
\end{tabular}
\end{center}

\bigskip
\section{Implementation of the oracle model}\label{sec:oracle-implementation}

In this section we present the sampling algorithm for the posterior covariance matrix distribution we
use for implementing the oracle Bayes rules.

\subsection{Sampling orthogonal matrices}

We begin by describing the subgroup algorithm of \citet{DiaconisShahshahani1987}
for generating orthogonal matrices from the Haar measure on $\opO (p)$ that is the basis for our
Gibbs sampler algorithms.
For $\rho \sim \text{Uniform} [0, 2 \pi]$ and $b\sim \text{Uniform}\{-1,+1\}$ independently,
 \[
  \vec{\Gamma}_2 = \begin{bmatrix}
\cos (\rho)  & \sin(\rho) \\
-b \sin (\rho)  & b \cos(\rho)
\end{bmatrix},
\]
is a random sample from the Haar measure on $\opO (2)$.
For $k = 3, \cdots, p$, let $\vec{\Gamma}_{k-1}$ be a sample from the Haar measure on $\opO (k-1)$.
Samples from the Haar measure on $\opO (k)$ are given by
 \begin{equation} \label{df1}
   \vec{\Gamma}_{k} =  h( \vec{U}_{k})
   \begin{bmatrix}
1         & 0 & \cdots & 0  \\
0         &    &           &     \\
\vdots  &    &    \vec{\Gamma}_{k-1}   &     \\
0         &     &                           &
\end{bmatrix}_{k \times k}.
\end{equation}
For the Householder matrix, $h( \vec{U}_k) = \vec{I}_k  - 2 \vec{v} \vec{v}^\top$,
where $\vec{v} = (\vec{e}_1 - \vec{U}_k) / \| \vec{e}_1 - \vec{U}_k \|$,
for unit vector  $\vec{e}_1  = (1, 0, \cdots, 0)^\top$
and $\vec{U}_k = \vec{Z} / \| \vec{Z} \|$ with $\vec{Z} \sim N_k( \vec{0}, \vec{I}_k)$.

In our implementation,  the sample from the Haar measure on $\opO (p)$
is the product of $ h( \vec{U}_{p})$ and the $p \times p$ matrices
constructed by embedding the Householder matrices for $k = p-1, \cdots, 3$, and $\vec{\Gamma}_2$
in the bottom-right of the  $p \times p$ identity matrix,
 \begin{align*}
 \vec{\Gamma} (\vec{U}_{p}, \vec{U}_{p-1}, \dots, \vec{U}_{3}, \rho, b)  =
h( \vec{U}_{p}) \times
   \begin{bmatrix}
1         & 0 & \cdots & 0  \\
0         &    &           &     \\
\vdots  &    &    h( \vec{U}_{p-1})   &     \\
0         &     &                           &
\end{bmatrix}_{p \times p} \hspace{-1.2em}
\times \;  \cdots \;
\times  \; \begin{bmatrix}
1         & 0 & \cdots & 0  \\
0         &    &           &     \\
\vdots  &    &   \vec{\Gamma}_2 ( \rho, b) &     \\
0         &     &                           &
\end{bmatrix}_{p \times p}.
\end{align*}
We denote by
$\vec{\Sigma} (\vec{U}_{p}, \vec{U}_{p-1}, \dots, \vec{U}_{3}, \rho, b; \vec{\Lambda}) =
  \bar{g}_{\vec{\Gamma}^\top (\vec{U}_{p},  \dots , b)}  ( \vec{\Lambda})$
the covariance matrix produced by this orthogonal matrix.

\subsection{Sampling from the oracle posterior} \label{sec:metrpls}

We employ the Gibbs sampler in Algorithm \ref{alg:oracle} to draw samples from
$\pi^{\opHaar} ( \vec{\Gamma} | \vec{x}, \vec{\lambda} )$.
At each Gibbs iteration, we carry out Metropolis updates for the components $\vec{U}_k$ for $k = p, \dots, 3$ as well as for $\rho$ and $b$.
In each Metropolis update, we generate a proposal for $\vec{U}_k$ of the form
$\vec{U}_k^* = \vec{W} / \| \vec{W} \|$, where $\vec{W} \sim N_k( \vec{U}_k, \sigma^2 \vec{I}_k)$, with
$\sigma^2$ tuned adaptively to achieve an adequate acceptance rate,
and then evaluate the probability
\[
\hbox{Pr}_k =
\min \left(\frac{f ( \vec{x} | \vec{\Sigma} ( \vec{U}_{p}, \dots, \vec{U}_{k+1}, \vec{U}_k^*, \vec{U}_{k-1},\dots, \vec{U}_{3}, \rho, b;  \vec{\Lambda}))}
{ f ( \vec{x} | \vec{\Sigma} ( \vec{U}_{p}, \dots, \vec{U}_{k+1}, \vec{U}_k, \vec{U}_{k-1},\dots, \vec{U}_{3}, \rho, b;  \vec{\Lambda}))}, 1 \right).
 \]
Then, with probability $\hbox{Pr}_k$ the value of $\vec{U}_k$ is updated to $\vec{U}_k^*$,
otherwise  the value of $\vec{U}_k$ remains unchanged.
In the Metropolis step for updating $\rho, b$,  the proposal value is
$\rho^* \sim N( \rho, \sigma^2 )$
with adaptively tuned $\sigma^2$ and $b^* \sim \text{Uniform} \{-1,+1\}$
and the updating probability is
\[
\hbox{Pr}_{\rho, b} =
\min \left(\frac{f ( \vec{x} | \vec{\Sigma} ( \vec{U}_{p}, \dots,  \vec{U}_{3}, \rho^*, b^*;  \vec{\Lambda}))}
{ f ( \vec{x} | \vec{\Sigma} ( \vec{U}_{p}, \dots, \vec{U}_{3}, \rho, b;  \vec{\Lambda}))}, 1 \right).
 \]

\begin{algorithm}[]
	\SetAlgoLined
	\SetKwInOut{Set}{Set}
	\SetKwInOut{Input}{Input}
	\SetKwInOut{Output}{Output}
	\Set {number of Gibbs sampler iterations $G$,  initial values of
    $\vec{U}_p^{(0)}, \dots, \vec{U}_3^{(0)}, \rho^{(0)}, b^{(0)}$ }
	\Input{$\vec{x}$, $\vec{\lambda}$}

	\vspace{.2cm}

	\For{$g = 1, \dots,  G$}{

		\For{$k = p,  \dots, 3$}{
			Perform Metropolis step for updating value of $\vec{U}_k^{(g)}$ }

		Perform Metropolis step for updating value of $\rho^{(g)}$ and $b^{(g)}$

		Compute $\vec{\Gamma}^{(g)} = \vec{\Gamma} ( \vec{U}_{p}^{(g)}, \dots, \vec{U}_{3}^{(g)},
        \rho^{(g)}, b^{(g)})$
	}
	\caption{Oracle covariance matrix Gibbs sampler}
	\label{alg:oracle}
\end{algorithm}

\medskip
In Section \ref{sec:oracle-eigenvalues-bayes-rules},
we presented
oracle Bayes rules of the form
$\hat{\vec{\Sigma}}(\vec{S}) = \vec{V}\hat{\vec{\Sigma}}(\vec{L})\vec{V}^\top$,
where $\hat{\vec{\Sigma}}(\vec{L}) = \opdiag(d_1,\dots,d_p)$
and each $d_k$ is a function of $\vec{\lambda}$ and of posterior expectations of functions of $\vec{\Gamma}$.
When $\vec{\lambda}$ is known, these Bayes rules can be computed by replacing the posterior expectations involving $\vec{\Gamma}$ with empirical means of the corresponding functions evaluated on posterior samples $\vec{\Gamma}^{(g)}$ generated by the oracle Gibbs sampler.

\bigskip
\section{Hierarchical Bayes modeling}
\label{sec:hbayes}

In this section, we present a hierarchical Bayes (hBayes) approach for approximating the oracle Bayes rules in the general case where $\vec{\lambda}$ is unobserved.

\subsection{The finite P\'olya tree model}

The hierarchical Bayes (hBayes) approach assumes that the components of $\vec{\lambda}$ are sampled from a finite P\'olya tree (FPT) model.
The FPT model generates distributions with piecewise constant density functions on a dyadic partition of ${\cal I}_0 = ( a_{\min}, a_{\max}]$,
corresponding to a fixed endpoints vector $\vec{a} = (a_{\min} = a_0 < a_1 <  \cdots < a_{2^K-1}  <  a_{2^K} = a_{\max})$.
The dyadic partition consists of subintervals
${\cal I}_{k,i} = (a_{(i-1) \cdot 2^{K-k}}, a_{i \cdot 2^{K-k}}]$,  for $k = 1 ,\dots, K$ and $i = 1, \dots, 2^k$.
The parameters of the FPT model are the Beta distribution hyper-parameters, $( \alpha_{k, j}, \beta_{k, j} )$ with $j = 1 ,\dots, 2^{k-1}$,
for the conditional subinterval probabilities $\phi_{k, j} \sim \text{Beta} ( \alpha_{k, j}, \beta_{k, j})$:
 $\Pr({\cal I}_{1, 1} | {\cal I}_0  ) = \phi_{1, 1}$,   $\Pr({\cal I}_{1,2} |  {\cal I}_0 )  = 1 - \phi_{1, 1}$.
 For $k = 2 ,\dots, K$, $\Pr({\cal I}_{k, 2 \cdot j - 1} | {\cal I}_{k-1, j}) = \phi_{k, j}$ and $\Pr({\cal I}_{k, 2 \cdot j } | {\cal I}_{k-1, j}) = 1 - \phi_{k, j}$.
 The subinterval probabilities, $\Pr({\cal I}_{k,j})  = \pi_{k, j}$,
are products of the conditional subinterval probabilities.
$\pi_{1, 1} =  \phi_{1, 1}$ and $\pi_{1, 2} = 1 - \phi_{1, 1}$.
For $k = 2, \dots ,K$,
$\pi_{k,2 \cdot  j -1} = \phi_{k,  j} \cdot \pi_{k-1,j}$ and $\pi_{k,2 \cdot  j} = (1- \phi_{k,  j}) \cdot \pi_{k-1,j}$.
 The level $K$ subinterval probabilities, $\vec{\pi}_K = (\pi_{K, 1} \cdots \pi_{K, 2^K})$, define the step function density:
\begin{equation} \label{def-step}
f ( \theta  |  \vec{\pi}_K  ; \vec{a} )  =  \pi_{K,1} \cdot  \frac{ I_{(a_0 , a_1]} ( \theta) }{a_1 - a_0} +
\pi_{K,2} \cdot \frac{ I_{(a_{1} , a_2]} ( \theta) }{ a_2 - a_1} + \cdots +
\pi_{K,2^K} \cdot \frac{ I_{(a_{2^K-1} , a_{2^K}]} ( \theta) }{ a_{2^K} - a_{2^K-1}},
\end{equation}
 for indicator function $I_{(a_{i-1}, a_i]} (\theta)$.

\bigskip
\subsection{The hierarchical Bayes generative model}
For computing hBayes estimates for $\vec{\Sigma}$ we assume the following
generative model for the FPT Beta random variables vector of $\boldsymbol{\phi}$,
the eigenvalue vector, the eigenvector matrix, and the data.

\begin{definition} {\bf The hierarchical Bayes generative model} \label{def:gen}
\begin{enumerate}
\item Generate $ f ( \theta |  \vec{\pi}_K  ; \vec{a} )$ from the FPT model with $\phi_{k,i} \sim \text{Beta}(1,1)$.
\item For $j = 1 ,\dots ,p$,  generate  $\theta_j \sim f (  \theta | \vec{\pi}_K ;  \vec{a})$.
\item Let $\vec{\Lambda} = \opdiag( \vec{\lambda})$, with
$\vec{\lambda} = \left( \theta_{(1)} >  \theta_{(2)} > \cdots > \theta_{(p)} \right)$, the values $\theta_1, \dots, \theta_p$ sorted in decreasing order.
\item Sample $\vec{\Gamma} \sim \pi^{\opHaar}$ and compute $\vec{\Sigma} = \vec{\Gamma}  \vec{\Lambda}  \vec{\Gamma}^\top$
\item For $i = 1, \dots, n$, sample  $\vec{X}^i \sim N_p( \vec{0}, \vec{\Sigma})$.
 \end{enumerate}
 \end{definition}

We use the Gibbs sampler in Algorithm~\ref{alg:hBayes} to generate posterior samples from the hBayes
generative model.
The Metropolis updates for generating $\vec{\Gamma}$ are run as described in Section \ref{sec:metrpls},
 where in this case $\vec{\lambda}$ is sampled in the algorithm and not a fixed oracle.
To sample $\vec{\lambda}$ given $\vec{\Gamma}$ and $\vec{\pi}_K$, we use the fact that the columns of
$\vec{x}\vec{\Gamma}$ are independent with $N(0, \lambda_j)$ components. With
$s_j = \|\vec{x}\vec{\Gamma}_j\|^2$, we draw $\lambda_1, \dots, \lambda_p$ independently, each from the density
proportional to $f(\lambda \mid \vec{\pi}_K; \vec{a})\, \lambda^{-n/2} \exp\{-s_j/(2\lambda)\}$, evaluated on a
grid over the FPT support $(a_{\min}, a_{\max}]$. We then sort the draws in decreasing order and permute
the columns of $\vec{\Gamma}$ accordingly.
For $k = 1 ,\dots, K$ and  $i = 1 ,\dots, 2^k$,
let $N_{k, i} ( \vec{\lambda})$ denote the number of $\lambda_j$ in subinterval ${\cal I}_{k, i}$.
According to the conjugacy of the FPT model,
the conditional distribution of the  Beta random variables given $\vec{\lambda}$ in the hBayes generative model is
$\phi_{k, j} | \vec{\lambda} \  \sim  \ \text{Beta}( 1 + N_{k, 2 \cdot j - 1}, 1 + N_{k, 2 \cdot j})$ for $j = 1 ,\dots, 2^{k-1}$.

\begin{algorithm}[]
	\SetAlgoLined
	\SetKwInOut{Set}{Set}
	\SetKwInOut{Input}{Input}
	\SetKwInOut{Output}{Output}
	\Set {number of Gibbs sampler iterations $G$,  initial values of $\vec{\pi}^{(0)}_K$, $\vec{\lambda}^{(0)}, \vec{U}_p^{(0)},  \dots, \vec{U}_3^{(0)}, \rho^{(0)}, b^{(0)}$ }
	\Input{$\vec{x}$}

	\vspace{.2cm}

	\For{$g = 1, \dots,  G$}{

		\For{$k = p,  \dots, 3$}{
			Perform Metropolis step for updating value of $\vec{U}_k^{(g)}$
            given $\vec{\lambda}^{(g-1)}$ }

		Perform Metropolis step for updating value of $\rho^{(g)}$ and $b^{(g)}$
        given $\vec{\lambda}^{(g-1)}$

        Compute $\vec{\Gamma}^{(g)} = \vec{\Gamma} ( \vec{U}_{p}^{(g)}, \dots, \vec{U}_{3}^{(g)},
        \rho^{(g)}, b^{(g)})$

		Sample value of $\vec{\lambda}^{(g)}$ from its conditional posterior distribution given
        $\vec{\Gamma}^{(g)}$ and  $\vec{\pi}^{(g-1)}_K$

		Sample values of $\boldsymbol{\phi}$ from their conditional posterior distribution given $\vec{\lambda}^{(g)}$ and
		use them to compute updated value of $\vec{\pi}^{(g)}_K$

	}
	\caption{hBayes covariance matrix Gibbs sampler}
	\label{alg:hBayes}
\end{algorithm}

The support of the FPT model distribution $(a_{\min}, a_{\max}]$ is chosen empirically so that it encompasses all positive sample eigenvalues. Due to the strong regularization embedded in the FPT model, the posterior draws of $f(\boldsymbol{\phi} \mid \vec{x})$ in the hBayes generative framework provide stable estimates of the marginal eigenvalue distribution.

By writing the posterior distribution of $\vec{\Gamma}$ and $\vec{\lambda}$ in the hBayes model as
\begin{equation} \label{eq:gen:pst}
\pi^{\ophBayes} ( \vec{\Gamma}, \vec{\lambda} \mid \vec{x} ) =
\int f ( \vec{\Gamma}, \vec{\lambda}, \boldsymbol{\phi} \mid \vec{x}) \, d \boldsymbol{\phi} =
\int f ( \vec{\Gamma} \mid  \vec{\lambda}, \vec{x})\,
f ( \vec{\lambda} \mid \boldsymbol{\phi}, \vec{x})\,
f ( \boldsymbol{\phi} \mid \vec{x}) \, d \boldsymbol{\phi},
\end{equation}
where $f ( \vec{\Gamma} \mid \vec{\lambda}, \vec{x})$ denotes the posterior distribution of $\vec{\Gamma}$ in the oracle model, we see that if the FPT posterior $f(\boldsymbol{\phi} \mid \vec{x})$ can successfully recover the marginal distribution of $\vec{\lambda}$—and thereby induce $f ( \vec{\lambda} \mid \boldsymbol{\phi}, \vec{x})$ that places high probability mass near the oracle $\vec{\lambda}$—then the resulting hBayes posterior for $\vec{\Sigma}$ can serve as a good approximation to the oracle posterior distribution of $\vec{\Sigma}$.

\subsection{Approximating the oracle Bayes rules}\label{sec:hbayes-rules}

We consider two estimators of $\vec{\Sigma}$ and $\vec{\Omega}$ based on the output
$(\vec{\Gamma}^{(g)}, \vec{\lambda}^{(g)})$, $g = 1, \dots, G$, of the hBayes Gibbs sampler.
The `hBayes' rules, the Bayes rules under the hBayes posterior $\pi^{\ophBayes}(\vec{\Gamma}, \vec{\lambda} \mid \vec{x})$
in (\ref{eq:gen:pst}), are given by Proposition~\ref{prop:oracle-bayes-rules} and Corollary~\ref{cor:precision-rules}, with the posterior expectations taken under
$\pi^{\ophBayes}$ (see the Remark following Corollary~\ref{cor:precision-rules}).
They are computed by averaging over the joint draws, with
$\tilde{\vec{\Sigma}}^{(g)} = \tilde{\vec{\Gamma}}^{(g)} \vec{\Lambda}^{(g)} \tilde{\vec{\Gamma}}^{(g)\top}$ and
$\tilde{\vec{\Omega}}^{(g)} = \tilde{\vec{\Gamma}}^{(g)} (\vec{\Lambda}^{(g)})^{-1} \tilde{\vec{\Gamma}}^{(g)\top}$,
where $\tilde{\vec{\Gamma}}^{(g)} = \vec{V}^\top \vec{\Gamma}^{(g)}$.
For the `eBayes' rules,  we estimate $\vec{\lambda}$ by $\hat{\vec{\lambda}}$, the vector of posterior medians of
the draws $\lambda^{(g)}_j$, and plug $\hat{\vec{\lambda}}$ into the oracle rules in place of $\vec{\lambda}$:
the expectations in Proposition~\ref{prop:oracle-bayes-rules} and Corollary~\ref{cor:precision-rules} are
replaced by averages of $\tilde{\vec{\Gamma}}^{(g)} \hat{\vec{\Lambda}}\, \tilde{\vec{\Gamma}}^{(g)\top}$ and
$\tilde{\vec{\Gamma}}^{(g)} \hat{\vec{\Lambda}}^{-1} \tilde{\vec{\Gamma}}^{(g)\top}$ over the draws
$\vec{\Gamma}^{(g)}$ of the hBayes Gibbs sampler.
Because the prior distribution of $(\vec{\lambda}, \vec{\Gamma})$ in the hBayes generative model is invariant
under $\opO(p)$, both rules are $\opO(p)$-equivariant estimators.
To see that the prior is invariant, note that $\vec{R} \in \opO(p)$ acts on the parameter by
$\vec{\Sigma} \mapsto \vec{R}^\top\vec{\Sigma}\vec{R}$, that is, it maps $\vec{\Gamma}$ to
$\vec{R}^\top\vec{\Gamma}$ and leaves $\vec{\lambda}$ unchanged. In Definition~\ref{def:gen},
$\vec{\Gamma}$ is Haar distributed independently of $\vec{\lambda}$, and $\vec{R}^\top\vec{\Gamma}$ is
again Haar distributed, so the prior of $(\vec{\lambda}, \vec{\Gamma})$ is invariant for any prior on
$\vec{\lambda}$, and in particular for the FPT prior. The FPT support $(a_{\min}, a_{\max}]$ is chosen from
the sample eigenvalues, which do not change when $\vec{S}$ is replaced by $\vec{R}^\top\vec{S}\vec{R}$, so
this data-dependent choice does not affect equivariance.
A variant, which we do not pursue here, runs the oracle Gibbs sampler with $\hat{\vec{\lambda}}$ in place of
$\vec{\lambda}$ and averages over its draws.

\section{Simulation results}
\label{sec:simulations}

In this section, we report simulation results that illustrate the behavior of the Gibbs sampling algorithms and compare the risks of the oracle and hBayes estimators with those of commonly used estimators.

We use the following design throughout this section.
We generated datasets with $p = 40$ and $n = 30$, $50$ and $100$ from $N_p(\mathbf 0, \boldsymbol\Sigma)$,
with $\boldsymbol\Sigma = \boldsymbol\Gamma\boldsymbol\Lambda\boldsymbol\Gamma^\top$, where $\boldsymbol\Gamma$ is a fixed
orthogonal matrix drawn once from the Haar measure on $\opO(p)$.
We considered two eigenvalue configurations. In the linear configuration, $\boldsymbol\lambda = (40, 39, \dots, 1)$.
In the clustered configuration, $\lambda_1, \dots, \lambda_{10}$ are equally spaced between
40 and 38, and $\lambda_{10}, \dots, \lambda_{40}$ are equally spaced between 38 and 1, giving
a block of ten nearly tied leading eigenvalues, which is known to be difficult to estimate.
For each of the six settings we simulated 100 datasets and ran the hBayes and oracle samplers for 2000 iterations,
discarding the first 500. The FPT has $K = 6$ levels, with equally spaced partition points on the
support $(0.9\, l_{\min}, 1.1\, l_1]$, where $l_1$ and $l_{\min}$ are the largest and smallest positive
sample eigenvalues. The same datasets and sampler runs are used in Sections~\ref{sec:sim-eigen}
and~\ref{sec:risk}. With the datasets analyzed in parallel on a desktop computer, the analysis of one
dataset, including both samplers and all the estimators of Section~\ref{sec:risk}, took 8 to 16 minutes,
depending on $n$.

\subsection{Gibbs sampler log-likelihood profiles}
\label{sec:sim-loglik}
Figure~\ref{fig1} shows the log-likelihood (\ref{norm-lik}) for one dataset from the linear configuration
with $n = 50$, with both samplers run for 5000 iterations.
The red line corresponds to $\log(f(\vec{x} \mid \vec{S}))$, and the blue line corresponds to $\log(f(\vec{x} \mid \vec{\Sigma}))$.
The black and green curves represent the log-likelihood trajectories $\log(f(\vec{x} \mid \vec{\Sigma}^{(g)}))$ for iterations $g = 1, \ldots, 5000$ of the oracle and hBayes samplers, respectively.
There is a substantial gap between the log-likelihood at $\vec{S}$, the MLE, and at $\vec{\Sigma}$.
The Gibbs sampler log-likelihood profiles are clustered around the log-likelihood for $\vec{\Sigma}$, indicating that the draws $\vec{\Sigma}^{(g)}$ tend to lie near $\vec{\Sigma}$ and away from $\vec{S}$.
Because the hBayes sampler updates both $\vec{\lambda}^{(g)}$ and $\vec{\Gamma}^{(g)}$, it yields more dispersed log-likelihood profiles than the oracle sampler, which is given $\vec{\lambda}$ and only samples $\vec{\Gamma}^{(g)}$.
In Appendix~\ref{app:cr}, we also check the oracle sampler against an exact posterior sampler for the
single-spike model.

\begin{figure}[t]
\centering
\includegraphics[width=0.85\textwidth]{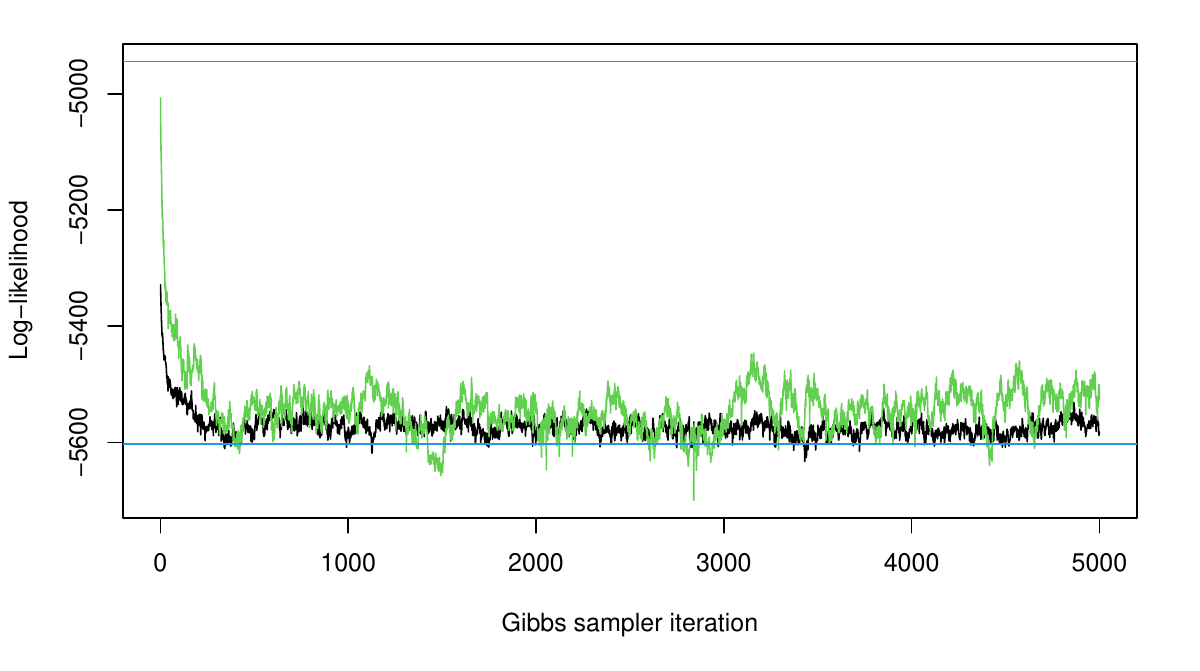}
\caption{Log-likelihood values in 5000 Gibbs sampler iterations. The black and green curves display the
log-likelihood in each iteration of the oracle sampler and hBayes sampler.
The red and blue horizontal lines display the
log-likelihood for the sample covariance matrix and for the true covariance matrix.}
\label{fig1}
\end{figure}

\subsection{Eigenvalue estimation in the hBayes sampler}
\label{sec:sim-eigen}

In this section, we examine how well the hBayes sampler estimates the eigenvalues.
We estimate the sorted eigenvalues by their posterior geometric means,
$\hat\lambda_j = \exp\{\mathrm E[\log\lambda_{(j)} \mid \mathbf x]\}$, which is the Bayes rule under
squared error on the log scale. We compare them with the sample eigenvalues $l_1, \dots, l_p$ and with
the QuEST estimates of \citet{LedoitWolf2015}, which invert the Mar\v{c}enko--Pastur equation to estimate
the population eigenvalues. We do not include the shrinkage estimators of Section~\ref{sec:risk} here,
because their eigenvalues are chosen to minimize a matrix loss, not to estimate the eigenvalues
of $\boldsymbol\Sigma$.

Figure~\ref{fig:fpt-cdf} shows the posterior of the FPT eigenvalue CDF for one dataset per setting.
Figure~\ref{fig3} shows the distribution of the estimates over the 100 datasets, Table~\ref{tab:eigen} summarizes their errors. The sample
eigenvalues are too spread out: they overestimate the large eigenvalues and underestimate the small ones.
The posterior median of the FPT CDF is close to the eCDF of the eigenvalues, and the 90\% posterior band
contains it except at the top, where the FPT places some mass above $\lambda_1$. The QuEST estimates have
little bias for $n \ge 50$ but a large variance, and for $n = 30$ they overestimate the ten smallest
eigenvalues by a factor of about 1.4. The hBayes estimates have a much smaller variance, which we see as
the effect of modeling the eigenvalues as draws from a common distribution that pools information across
the spectrum. They overestimate the two or three largest eigenvalues and, for $n \le 50$, the one or two
smallest; for $n = 30$ they underestimate the others by up to 20\%, and for $n = 50$ by about half as
much. In Table~\ref{tab:eigen}, hBayes has the smallest error in every setting except for the ten
largest eigenvalues in the clustered configuration with $n = 30$, where QuEST is slightly better.

In our view, the upward bias of the leading hBayes estimates reflects the likelihood. When $p/n$ is
large, the largest sample eigenvalues overestimate $\lambda_1, \lambda_2, \dots$ substantially, and the
likelihood strongly supports leading eigenvalues that are larger than the true ones. The bias decreases
as $n$ increases, and other settings of the FPT did not reduce it consistently.

\begin{figure}[p]
\includegraphics[width=1.00\textwidth,height=0.56\textwidth]{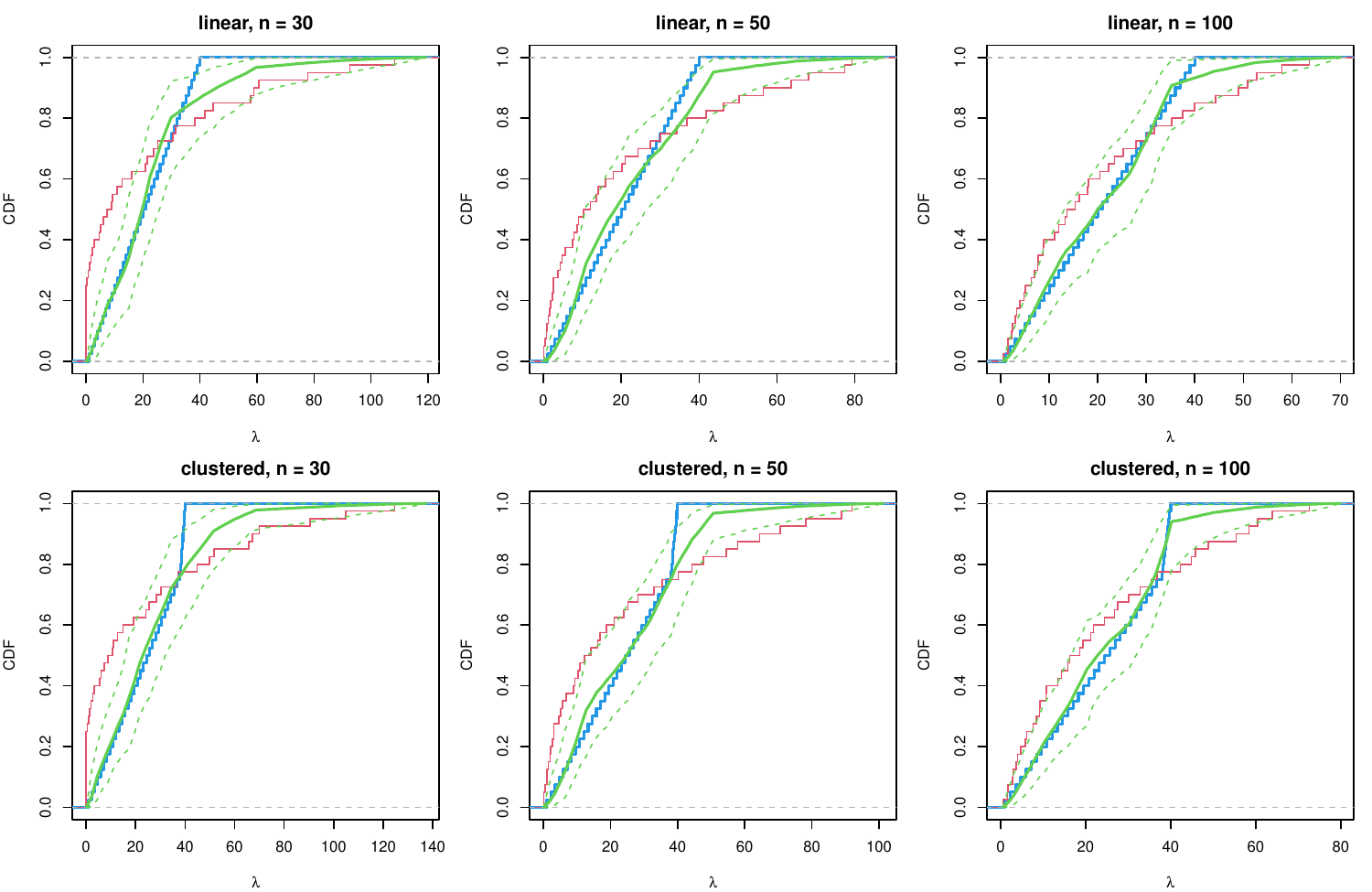}
\caption{Estimation of the eigenvalue distribution by the FPT for one simulated dataset per setting:
linear (top) and clustered (bottom) configuration, $n = 30$, $50$ and $100$ (left to right). Blue: eCDF
of the eigenvalues $\lambda_1, \dots, \lambda_p$; red: eCDF of the sample eigenvalues (for $n = 30$,
$l_{31} = \cdots = l_{40} = 0$); green: posterior median (solid) and 0.05 and 0.95 posterior quantiles
(dashed) of the FPT eigenvalue CDF over iterations $g = 501, \dots, 2000$.}
\label{fig:fpt-cdf}

\vspace{0.6cm}
\includegraphics[width=1.00\textwidth,height=0.56\textwidth]{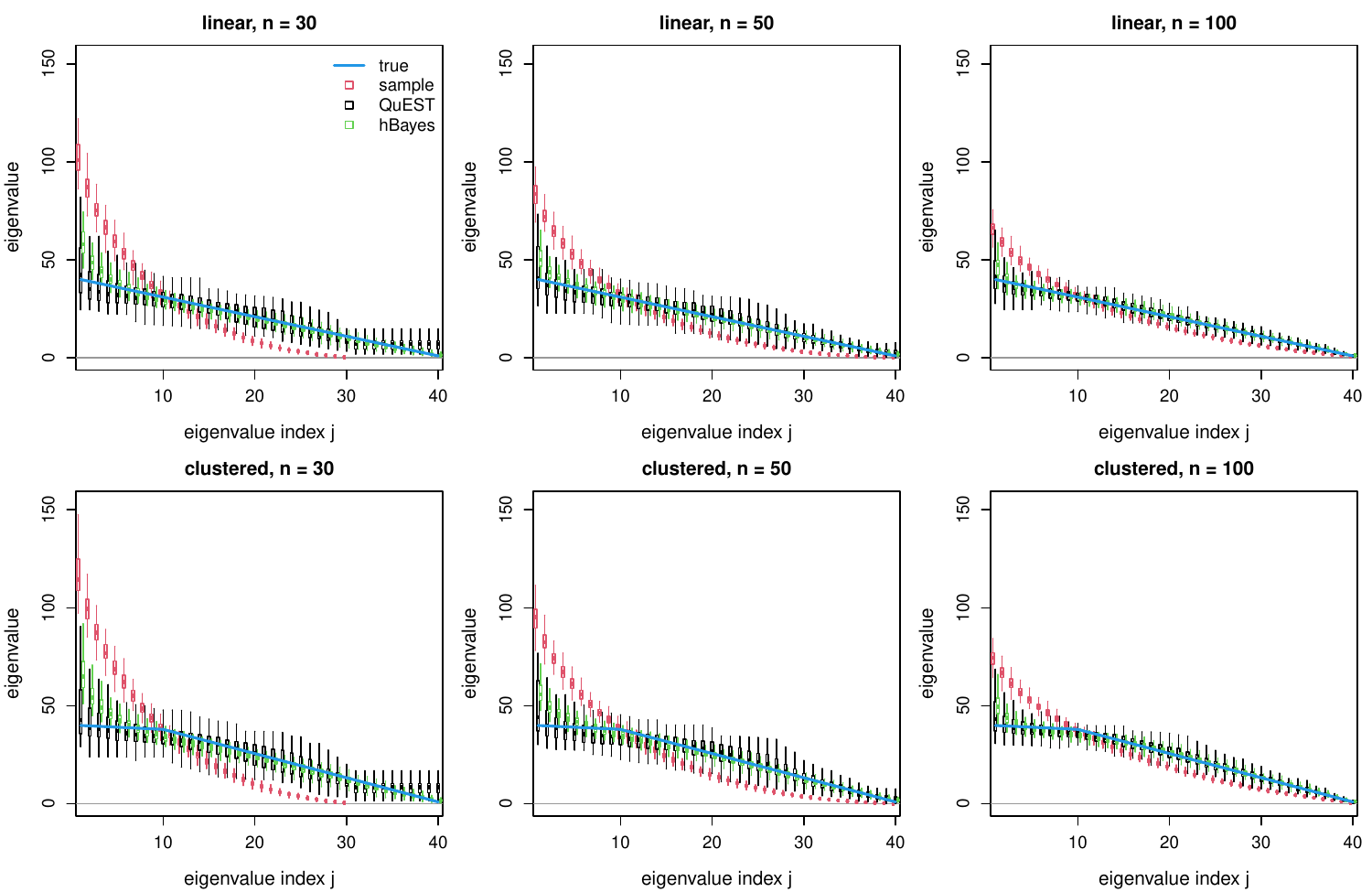}
\caption{Eigenvalue estimates in 100 simulated datasets per setting, arranged as in
Figure~\ref{fig:fpt-cdf}. At each sorted eigenvalue index $j$, the boxplots show, from left to right,
the sample eigenvalues (red), the QuEST estimates (black) and the hBayes posterior geometric means
(green); the blue curve shows the eigenvalues $\lambda_j$. For $n = 30$, the sample eigenvalues are
shown for $j \le 30$ only, as $l_{31} = \cdots = l_{40} = 0$.}
\label{fig3}
\end{figure}

\begin{table}[ht]
\centering
\small
\begin{tabular}{lrrrrrr}
\hline
 & \multicolumn{3}{c}{Linear} & \multicolumn{3}{c}{Clustered} \\
\cmidrule(lr){2-4}\cmidrule(lr){5-7}
 & All & Top 10 & Bottom 30 & All & Top 10 & Bottom 30 \\
\hline
\multicolumn{7}{l}{$n = 30$} \\
Sample  & --     & 3.006 & --     & --     & 3.877 & --     \\
QuEST   & 10.479 & 0.549 & 9.930  & 11.759 & 0.520 & 11.239 \\
hBayes  & 1.974  & 0.316 & 1.658  & 2.635  & 0.543 & 2.092  \\
\multicolumn{7}{l}{$n = 50$} \\
Sample  & 41.507 & 1.815 & 39.692 & 42.537 & 2.460 & 40.077 \\
QuEST   & 4.595  & 0.408 & 4.186  & 4.841  & 0.425 & 4.416  \\
hBayes  & 1.275  & 0.143 & 1.132  & 1.536  & 0.263 & 1.273  \\
\multicolumn{7}{l}{$n = 100$} \\
Sample  & 7.205  & 0.805 & 6.399  & 7.798  & 1.206 & 6.592  \\
QuEST   & 1.072  & 0.172 & 0.899  & 1.171  & 0.177 & 0.994  \\
hBayes  & 0.456  & 0.085 & 0.371  & 0.526  & 0.124 & 0.402  \\
\hline
\end{tabular}
\caption{Log-scale squared error $\sum_j \{\log(\hat\lambda_j/\lambda_j)\}^2$ of the eigenvalue estimates,
summed over all 40 eigenvalues, the 10 largest and the 30 smallest, and averaged over 100 datasets;
$p = 40$. For $n = 30$ the ten smallest sample eigenvalues are zero.}
\label{tab:eigen}
\end{table}

\subsection{Covariance and precision matrix estimation}\label{sec:risk}

In this section, we compare the risks of the Bayes rules with those of commonly used estimators of the
covariance and precision matrices. For each dataset we computed three Bayes rules.
The oracle rule uses the true $\vec{\lambda}$ and draws from
$\pi^{\opHaar}(\vec{\Gamma}\mid\vec{x},\vec{\lambda})$. The hBayes rule is computed from the joint
draws $(\vec{\Gamma}^{(g)},\vec{\lambda}^{(g)})$ of the hBayes Gibbs sampler, and the eBayes rule
is the oracle rule with $\vec{\lambda}$ replaced by the posterior medians of the hBayes eigenvalue draws,
as described in Section~\ref{sec:hbayes-rules}.
All three are computed as in Proposition~\ref{prop:oracle-bayes-rules} and Appendix~\ref{app:precision}, with the last $p-n$
entries tied when $n<p$. Each rule is evaluated under the loss it minimizes, that is, under $L_0$, $L_1$ and $L_2$
for $\vec{\Sigma}$ and under the corresponding losses for $\vec{\Omega}$.

Table~\ref{tab:competitors} lists the competitors. They were run with their default tuning. Since the
mean is known to be zero, the Ledoit--Wolf estimators are computed without centering, and the
competitors' precision estimates are the inverses of their covariance estimates, except for the
graphical lasso, which is designed for sparse $\vec{\Omega}$; this does not hold here, and we include it
because it is widely used. For $n=30<p$, the sample covariance matrix is singular and JS, Haff and Stein
are not defined. All competitors except JS and the graphical lasso are orthogonally equivariant
(Section~\ref{sec:equivariance-to-O(p)}), so by Proposition~\ref{prop3} their risks cannot be smaller than that of
the oracle rule.

\begin{table}[ht]
\centering
\small
\begin{tabular}{llll}
\hline
Estimator & Description & Reference & Equivariant \\
\hline
JS & minimax estimator & \citet{JamesStein1961} & no \\
Haff & empirical Bayes estimator & \citet{Haff1980} & yes \\
Stein & Stein's estimator, isotonized & \citet{Stein1975, LinPerlman1985} & yes \\
LW.lin & linear shrinkage & \citet{LedoitWolf2004} & yes \\
LW.nl & nonlinear shrinkage with QuEST & \citet{LedoitWolf2012, LedoitWolf2015} & yes \\
glasso & graphical lasso, 5-fold CV penalty ($\vec{\Omega}$ only) & \citet{friedman2008sparse} & no \\
\hline
\end{tabular}
\caption{Competing estimators in Section~\ref{sec:risk}, in addition to the sample covariance matrix.
Haff and Stein are computed with the \texttt{R} package \texttt{stcov} \citep{stcov_man}, LW.nl and
QuEST with \texttt{nlshrink} \citep{nlshrink_man}, and the graphical lasso with \texttt{glasso}
\citep{glasso_man}.}
\label{tab:competitors}
\end{table}

Table~\ref{tab:risk} reports the mean losses for the linear configuration; Table~\ref{tab:risk-full} in
Appendix~\ref{app:risk-full} gives the results for both configurations, including the sample covariance
matrix and JS. For $n \ge 50$ the Bayes rules have the smallest mean loss under all six losses,
and the gains are largest under the losses that weight the small eigenvalues: for $n = 50$, under
$L_2$ for $\vec{\Sigma}$ the hBayes risk is about half that of the best competitor, Haff, and under $L_2$
for $\vec{\Omega}$ it is about 25\% below that of LW.lin. Although Proposition~\ref{prop3} does not cover
JS and the graphical lasso, their mean losses exceed that of the oracle rule by at least 60\% and 20\%,
respectively, in every setting and under every loss. The gap between the oracle and hBayes rules
measures the cost of not knowing $\vec{\lambda}$; under $L_1$ for $\vec{\Sigma}$ it is about 8\% for
$n = 30$ and under 2\% for $n = 100$. For $n=100$, the hBayes and eBayes rules have mean $L_2$ loss for
$\vec{\Sigma}$ slightly below that of the oracle rule (8.31 and 8.37 versus 8.40; for hBayes, paired
difference $-0.084$, standard error $0.023$). We attribute this to Monte Carlo error in the posterior
expectations of Proposition~\ref{prop:oracle-bayes-rules}, in particular in
$\mathrm E[\tilde{\vec{\Omega}}\circ\tilde{\vec{\Omega}}]$, which the $L_2$ rule requires.

The Ledoit--Wolf estimators are competitive under the Frobenius loss $L_0$ for $\vec{\Sigma}$. For
$n = 30$, LW.lin and LW.nl have slightly smaller mean $L_0$ loss than the hBayes rule, although the eBayes
rule is better than both; for $n = 50$ hBayes improves on LW.nl by less than 2\%, and for $n = 100$ the
two are within one standard error. Under the other losses, the Bayes rules are better than the
Ledoit--Wolf estimators for every $n$, except that for $n = 30$ hBayes ties LW.lin under $L_1$ for
$\vec{\Omega}$. LW.lin does poorly under $L_1$ and $L_2$ for $\vec{\Sigma}$. These losses weight errors in
the directions of the small eigenvalues by $\vec{\Sigma}^{-1}$, and linear shrinkage pulls all eigenvalues
toward their mean, which is about 20, so the smallest eigenvalues, with $\lambda_{40}=1$, are greatly
overestimated.

The eBayes rule has mean loss no larger than the hBayes rule under $L_0$ for $\vec{\Sigma}$ and under
$L_1$ and $L_2$ for $\vec{\Omega}$, and the hBayes rule is better under the other three losses for
$n \ge 50$. The rules in the first group are dominated by the leading eigenvalues, whose posterior is
shifted upward (Section~\ref{sec:sim-eigen}), and replacing the eigenvalue draws by their medians removes
part of this shift. The eBayes rule is, however, less stable under $L_2$ for $\vec{\Sigma}$: for
$n\le 50$ its standard errors are two to five times those of hBayes, because a few datasets have very
large losses. The results for the clustered configuration are similar, except that for $n = 100$ the
estimator of Haff is slightly better than hBayes under $L_0$ for $\vec{\Omega}$, by less than one
standard error.

\begin{table}[ht]
\centering
\small
\setlength{\tabcolsep}{4pt}
\begin{tabular}{llrrrrrr}
\hline
 & & \multicolumn{3}{c}{$\boldsymbol\Sigma$} & \multicolumn{3}{c}{$\boldsymbol\Omega$} \\
\cmidrule(lr){3-5}\cmidrule(lr){6-8}
 & & $L_0$ & $L_1$ & $L_2$ & $L_0$ & $L_1$ & $L_2$ \\
\hline
$n=30$ & oracle & 4174 (15) & 13.76 (0.05) & 17.04 (0.09) & 0.999 (0.002) & 9.10 (0.02) & 9.16 (0.01) \\
 & eBayes & \textbf{4390 (29)} & \textbf{14.67 (0.10)} & 20.75 (0.50) & \textbf{1.078 (0.015)} & \textbf{9.33 (0.04)} & \textbf{9.46 (0.03)} \\
 & hBayes & 4537 (37) & 14.80 (0.13) & \textbf{17.57 (0.11)} & 1.191 (0.038) & 9.39 (0.04) & 9.68 (0.03) \\
 & LW.lin & 4475 (21) & 28.37 (0.22) & 341.04 (4.89) & 1.243 (0.002) & 9.40 (0.02) & 12.89 (0.16) \\
 & LW.nl & 4480 (32) & 26.63 (0.34) & 299.60 (7.47) & 1.223 (0.004) & 9.43 (0.04) & 13.87 (0.30) \\
 & glasso & -- & -- & -- & 1.259 (0.002) & 11.58 (0.04) & 19.04 (0.31) \\
\hline
$n=50$ & oracle & 3492 (15) & 9.50 (0.06) & 12.87 (0.07) & 0.701 (0.012) & 7.77 (0.03) & 8.76 (0.02) \\
 & eBayes & \textbf{3587 (19)} & 10.32 (0.10) & 15.69 (0.54) & 0.835 (0.009) & \textbf{7.95 (0.04)} & \textbf{8.91 (0.03)} \\
 & hBayes & 3656 (23) & \textbf{10.01 (0.08)} & \textbf{13.43 (0.22)} & \textbf{0.808 (0.011)} & 7.98 (0.04) & 9.01 (0.03) \\
 & Haff & 4666 (32) & 10.34 (0.09) & 24.28 (0.84) & 1.536 (0.143) & 19.62 (0.59) & 156.48 (11.82) \\
 & Stein & 4335 (42) & 11.77 (0.19) & 42.67 (1.87) & 1.506 (0.143) & 15.91 (0.56) & 117.46 (11.46) \\
 & LW.lin & 3925 (16) & 24.95 (0.16) & 275.90 (3.27) & 1.216 (0.002) & 8.65 (0.02) & 11.85 (0.11) \\
 & LW.nl & 3718 (24) & 17.18 (0.26) & 125.51 (4.63) & 1.134 (0.066) & 8.75 (0.24) & 22.57 (4.37) \\
 & glasso & -- & -- & -- & 1.240 (0.002) & 10.52 (0.04) & 16.35 (0.22) \\
\hline
$n=100$ & oracle & 2513 (12) & 5.23 (0.03) & 8.40 (0.05) & 0.244 (0.007) & 4.94 (0.03) & 7.18 (0.03) \\
 & eBayes & \textbf{2542 (14)} & 5.33 (0.03) & 8.37 (0.04) & 0.318 (0.006) & \textbf{5.05 (0.03)} & \textbf{7.30 (0.03)} \\
 & hBayes & 2564 (15) & \textbf{5.31 (0.03)} & \textbf{8.31 (0.04)} & \textbf{0.305 (0.006)} & \textbf{5.05 (0.03)} & 7.33 (0.03) \\
 & Haff & 3036 (19) & 5.41 (0.03) & 11.24 (0.12) & 0.312 (0.011) & 7.16 (0.07) & 24.88 (0.45) \\
 & Stein & 2928 (22) & 5.92 (0.07) & 15.34 (0.49) & 0.321 (0.011) & 6.54 (0.08) & 20.20 (0.54) \\
 & LW.lin & 3034 (12) & 18.89 (0.10) & 175.92 (1.70) & 1.154 (0.001) & 7.18 (0.02) & 9.83 (0.07) \\
 & LW.nl & 2572 (15) & 6.95 (0.05) & 24.32 (0.32) & 0.435 (0.011) & 5.13 (0.03) & 10.34 (0.14) \\
 & glasso & -- & -- & -- & 1.142 (0.003) & 8.46 (0.03) & 14.12 (0.14) \\
\hline
\end{tabular}
\caption{Mean loss (standard error) over 100 datasets, linear configuration, $p=40$. The oracle rule
uses the true eigenvalues. Bold: smallest mean loss among the non-oracle estimators. For $n=30<p$, JS,
Haff and Stein are not defined; the graphical lasso estimates only $\boldsymbol\Omega$.}
\label{tab:risk}
\end{table}

\section{Data example: industry portfolio returns}
\label{sec:data}

In this section, we apply the hBayes and eBayes rules to monthly returns of US industry portfolios and
compare them with covariance estimators used in the finance literature.

\subsection{Data}

We use the monthly returns, expressed in percent, of the 49
value-weighted industry portfolios of Fama and French.
Each portfolio represents an industry category and contains
stocks weighted in proportion to their market capitalization.
The data are obtained from the Kenneth R. French Data Library
\citep{French2026} and are based on the August 2026 release of the
database of the Center for Research in Security Prices (CRSP).
Returns are available for all 49 portfolios from July 1969,
giving a sample of $686$ months, from July 1969 to August 2026.
In the full sample, the largest eigenvalue
of the sample covariance matrix accounts for 55\% of the total variance, and the condition number of the
sample covariance matrix is about 440.

\subsection{Design}

Each July, from 1974 for $n=60$ and from 1979 for $n=120$, we apply each
estimator to the preceding $n$ months of returns
($p/n\approx0.82$ and $0.41$, respectively),
calculate the corresponding portfolio weights, and hold them fixed
for the following twelve months.
For $n=60$, this gives 53 training windows and 626 test months, and
for $n=120$, it gives 48 training windows and 566 test months,
ending in August 2026.

Because the mean is unknown, each training window is centered at its sample mean. The methods of
Sections~\ref{sec:oracle-eigenvalues-bayes-rules}--\ref{sec:hbayes} depend on the data only through
$n\vec{S} = \vec{X}^\top\vec{X}$, and the centered cross-product matrix has the distribution of
$\vec{X}^\top\vec{X}$ for $n-1$ independent $N_p(\vec{0},\vec{\Sigma})$ rows, so we apply them with
$n\vec{S}$ replaced by the centered cross-product matrix and $n$ by $n-1$ (equivalently, to
$\vec{Z} = \vec{H}\vec{X}$, where $\vec{H}$ is the $(n-1)\times n$ matrix of normalized Helmert
contrasts).

The hBayes and eBayes estimates of $\vec{\Omega}$ are the Bayes rules under the Stein loss
$L_1(\hat{\vec{\Omega}}, \vec{\Omega})$ of Corollary~\ref{cor:precision-rules},
$\omega_k = 1/\mathrm E[\tilde\Sigma_{kk} \mid \vec{x}]$.
They are computed as in Section~\ref{sec:hbayes-rules} from one run of the hBayes Gibbs sampler per window,
with the settings of Section~\ref{sec:simulations}.
Table~\ref{tab:estimators} lists the competing estimators; LW.nl is computed as in Section~\ref{sec:risk}.
The estimators that are not orthogonally equivariant are commonly used for asset returns and use structure
tied to the coordinates (Section~\ref{sec:equivariance-to-O(p)}). For LW.SI, the market proxy is the
equally weighted average of the 49 portfolios. LW.SI, LW.CC and NoShort are applied to the uncentered
window with the authors' formulas for an unknown mean, and the other estimators to the centered window
with $n-1$ degrees of freedom. The precision estimates of the competitors are the inverses of their
covariance estimates.

\begin{table}[ht]
\centering
\small
\begin{tabular}{lll}
\hline
Estimator & Description & Reference \\
\hline
\multicolumn{3}{l}{\emph{Orthogonally equivariant}} \\
hBayes, eBayes & Bayes rules of Section~\ref{sec:hbayes-rules} & this paper \\
LW.nl & nonlinear shrinkage with QuEST & \citet{LedoitWolf2012, LedoitWolf2015} \\
LW.anl & analytical nonlinear shrinkage & \citet{LedoitWolf2020} \\
LW.lin & linear shrinkage & \citet{LedoitWolf2004} \\
Sample & sample covariance matrix & \\
\multicolumn{3}{l}{\emph{Not orthogonally equivariant}} \\
LW.SI & shrinkage toward the single-index model & \citet{LedoitWolf2003} \\
LW.CC & shrinkage toward the constant-correlation model & \citet{LedoitWolf2004jpm} \\
PCA1, PCA3 & factor models with diagonal residual covariance & \citet{Chan1999, Fan2013} \\
NoShort & minimum variance without short sales & \citet{JagannathanMa2003} \\
\multicolumn{3}{l}{\emph{Reference}} \\
EW & equally weighted portfolio & \citet{DeMiguel2009} \\
\hline
\end{tabular}
\caption{Estimators compared in Section~\ref{sec:data}.}
\label{tab:estimators}
\end{table}

We compare the estimators by two out-of-sample scores. The first is the volatility of the global
minimum-variance portfolio, which is the usual score for covariance estimators in portfolio selection
\citep{LedoitWolf2017}. For an estimate $\hat{\vec{\Omega}}$, the global minimum-variance portfolio has
weights $\vec{w} = \hat{\vec{\Omega}} \mathbf{1} / (\mathbf{1}^\top \hat{\vec{\Omega}} \mathbf{1})$. We
summarize its realized returns $\vec{w}^\top \vec{x}_t$ over the test months by their annualized standard
deviation (SD). The differences in SD between each competitor and hBayes are
tested with the HAC test of \citet{LedoitWolf2011} for the difference of log variances, using a
Newey--West kernel \citep{NeweyWest1987}.

The minimum-variance score depends on $\hat{\vec{\Omega}}$ only through $\hat{\vec{\Omega}}\mathbf{1}$, and
it is not invariant under $\opO(p)$ because the budget constraint singles out the direction $\mathbf{1}$.
The second score is the held-out Gaussian log-likelihood of each test month,
\begin{equation*}
  \ell_t(\hat{\vec{\Omega}}) = \tfrac12 \log\det\hat{\vec{\Omega}}
    - \tfrac12 (\vec{x}_t - \bar{\vec{x}})^\top \hat{\vec{\Omega}} (\vec{x}_t - \bar{\vec{x}})
    - \tfrac{p}{2}\log(2\pi),
\end{equation*}
where $\vec{x}_t$ is the vector of returns in test month $t$ and $\bar{\vec{x}}$ is the mean of the
training window. This score is $\opO(p)$-invariant. If the test month has covariance
$\vec{\Sigma}$, the expectation of $\ell_t$ is $-\tfrac12 L_1(\hat{\vec{\Omega}}, \vec{\Omega})$ up to an
additive constant and
a term of order $1/n$ from estimating the mean. The average of $\ell_t$ over the test months is therefore
an out-of-sample estimate of the risk studied in Section~\ref{sec:risk}. We report its difference from
hBayes, with a Newey--West standard error.

The precision rule we use is also the Bayes rule for the minimum-variance problem. For fixed weights
$\vec{w}$, the posterior expected variance of the portfolio return is
$\vec{w}^\top \mathrm E[\vec{\Sigma} \mid \vec{x}]\, \vec{w}$, which is minimized subject to
$\mathbf{1}^\top \vec{w} = 1$ by $\vec{w} \propto \mathrm E[\vec{\Sigma}\mid \vec{x}]^{-1}\mathbf{1}$. Under an
orthogonally invariant prior,
$\mathrm E[\vec{\Sigma} \mid \vec{x}] = \vec{V} \opdiag(\mathrm E[\tilde\Sigma_{kk} \mid \vec{x}]) \vec{V}^\top$, so this is
the minimum-variance portfolio of the $L_1$ precision rule.

\subsection{Results}

Table~\ref{tab:portfolio} reports the results. Among the orthogonally equivariant estimators, hBayes and
eBayes have the lowest SD and the highest log-likelihood for both window lengths, and they differ from each
other by at most 0.01 in both scores. Most of the differences from the other equivariant estimators are
significant at the 5\% level.

Among the estimators that are not orthogonally equivariant, LW.SI, LW.CC and NoShort for $n = 60$ and PCA3
for $n = 120$ have lower SD than hBayes, but none of these differences is significant; the $p$-values are
0.13 for LW.SI with $n = 60$ and 0.75 for PCA3 with $n = 120$. All of these estimators except NoShort, for
which the log-likelihood is not computed, have a lower log-likelihood than hBayes. PCA3 has the lowest SD for
$n = 120$, but its log-likelihood is 1.27 below that of hBayes. Its three-factor structure appears to
estimate the direction $\vec{\Omega}\mathbf{1}$ that determines the portfolio well, while fitting the rest of the
precision matrix poorly. These comparisons are based on a single dataset, the tests are not adjusted for
multiple comparisons, and all the methods treat the covariance matrix as constant within each training
window.

Figure~\ref{fig:cumloglik} shows the cumulative log-likelihood of several estimators minus that of LW.nl
over the test period. For $n = 60$ the hBayes curve gains mainly in two episodes, in 1983 and in the 2008--09 financial
crisis, and is nearly flat in between; for $n = 120$ it increases slowly throughout the test period.
The curve for LW.CC rises in calm periods and falls sharply around 2000 and 2020, and, for $n = 120$,
also in 2009.

We also compared, for each training window of $n = 60$ months, the hBayes estimates of the eigenvalues
(posterior geometric means, Section~\ref{sec:sim-eigen}) with the sample eigenvalues of the window and
with the eigenvalues of the sample covariance matrix of the full period, for which $p/n \approx 0.07$. For
the ten largest eigenvalues, the hBayes estimates are close to the sample eigenvalues. For the small
eigenvalues the two differ markedly. Relative to the smallest full-period eigenvalue, the median over
windows of the smallest sample eigenvalue is 0.02, and that of the smallest hBayes estimate is 0.30. The small eigenvalues of the full-period
covariance matrix may themselves be inflated, because it averages over periods with different covariance
structures. Both the minimum-variance portfolio and the held-out likelihood depend mainly on the small
eigenvalues, and the differences in Table~\ref{tab:portfolio} and Figure~\ref{fig:cumloglik} are
consistent with the hBayes estimates correcting the downward bias of the small sample eigenvalues.

\begin{table}[p]
\centering
\small
\begin{tabular}{l rr rr}
\hline
 & \multicolumn{2}{c}{$n = 60$} & \multicolumn{2}{c}{$n = 120$} \\
Estimator & SD & $\Delta\ell$ & SD & $\Delta\ell$ \\
\hline
\multicolumn{5}{l}{\emph{Orthogonally equivariant}} \\
hBayes & 12.71 & 0 & 12.37 & 0 \\
eBayes & 12.71 & $-0.01$ (0.01) & 12.37 & 0.00 (0.01) \\
LW.nl & 12.74 & $-0.33$ (0.16)$^{\dagger}$ & 12.49$^{*}$ & $-0.07$ (0.04) \\
LW.anl & 12.92$^{*}$ & $-1.59$ (0.27)$^{\dagger}$ & 12.45$^{*}$ & $-0.07$ (0.03)$^{\dagger}$ \\
LW.lin & 12.88$^{*}$ & $-2.37$ (0.60)$^{\dagger}$ & 12.71$^{*}$ & $-1.38$ (0.34)$^{\dagger}$ \\
Sample & 25.32$^{*}$ & $-176.78$ (12.59)$^{\dagger}$ & 15.15$^{*}$ & $-16.41$ (1.66)$^{\dagger}$ \\
\multicolumn{5}{l}{\emph{Not orthogonally equivariant}} \\
LW.SI & 12.50 & $-0.34$ (0.37) & 12.41 & $-0.23$ (0.20) \\
LW.CC & 12.69 & $-0.15$ (0.28) & 12.57 & $-0.55$ (0.25)$^{\dagger}$ \\
PCA1 & 13.10 & $-1.35$ (0.40)$^{\dagger}$ & 13.10 & $-2.11$ (0.31)$^{\dagger}$ \\
PCA3 & 13.28$^{*}$ & $-1.72$ (0.46)$^{\dagger}$ & 12.27 & $-1.27$ (0.26)$^{\dagger}$ \\
NoShort & 12.68 & -- & 12.42 & -- \\
\multicolumn{5}{l}{\emph{Reference}} \\
EW & 16.88$^{*}$ & -- & 16.45$^{*}$ & -- \\
\hline
\end{tabular}
\caption{Out-of-sample performance of global minimum-variance portfolios of the
49 Fama--French industry portfolios, July 1974 ($n = 60$) or July 1979
($n = 120$) to August 2026. Weights are estimated from the previous $n$ months
and held for 12 months. SD: annualized standard deviation of monthly portfolio
returns (\%); $\Delta\ell$: mean
held-out Gaussian log-likelihood per test month minus that of hBayes, with HAC
standard error (negative values favor hBayes). $^{*}$: SD differs from that of
hBayes at the 5\% level (Ledoit--Wolf HAC test); $^{\dagger}$: $\Delta\ell$
differs from zero at the 5\% level.}
\label{tab:portfolio}
\end{table}

\begin{figure}[p]
\centering
\includegraphics[width=\textwidth]{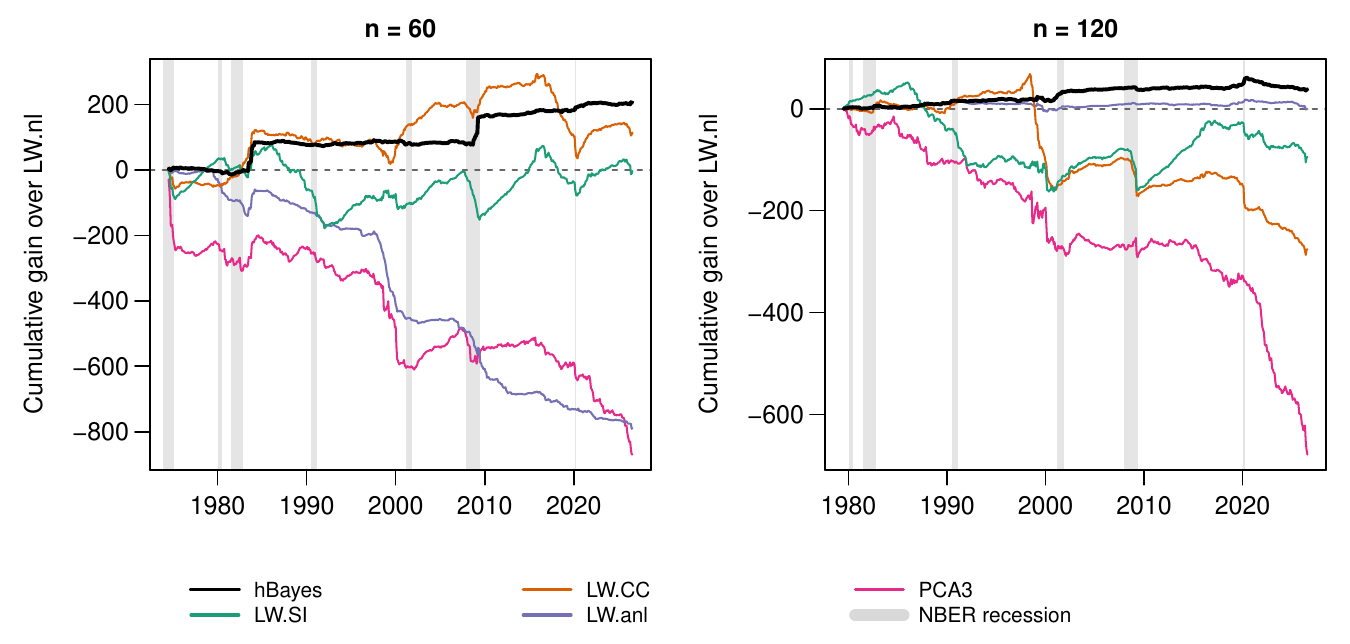}
\caption{Cumulative held-out Gaussian log-likelihood gain over LW.nl for the
49 industry portfolios, with training windows of $n = 60$ (left) and
$n = 120$ (right) months re-estimated annually. Each curve is the running sum,
over test months, of the log-likelihood of the estimator minus that of LW.nl;
an increasing curve indicates a better fit than LW.nl. The baseline is LW.nl, one of the two best
equivariant competitors in Table~\ref{tab:portfolio}, so that the hBayes curve itself is shown. eBayes is
indistinguishable from hBayes and is omitted. Shaded: NBER recessions.}
\label{fig:cumloglik}
\end{figure}
\section{Discussion}
\label{sec:discussion}

For any scientific question that can be encoded by an orthogonally invariant loss, the risk of the oracle
Bayes rule is a lower bound on the risk of every equivariant estimator (Proposition~\ref{prop3}), and the
hBayes posterior draws provide an approximation to this rule. The three
losses studied here are examples. The hBayes estimators gained most under the losses that weight the small
eigenvalues, the Stein and squared Stein losses for $\vec{\Sigma}$ and the losses for $\vec{\Omega}$, and
when $p/n$ was large. Under the Frobenius loss for $\vec{\Sigma}$, nonlinear shrinkage did about as well.

When $p/n$ is large, the likelihood strongly supports leading eigenvalues that are
larger than the true ones, and the weakly informative FPT prior cannot correct the resulting upward bias
(Section~\ref{sec:sim-eigen}). For small $n$, reducing this bias requires domain knowledge encoded in an
informative prior for the eigenvalue distribution. 
Further, because nothing is assumed about the
eigenvectors, the sampler updates all $p$ of them, with one fit taking about six minutes for $p = 49$. A
partially informative prior that sets the smallest eigenvalues equal, as in the spiked covariance model
\citep{donoho2018optimal}, makes the posterior invariant to rotations of their eigenvectors
(Lemma~\ref{lemm4a}), so only the $k$ leading eigenvectors would need to be sampled.
Appendix~\ref{app:spike} gives the case $k = 1$.

The same generality means that the approach needs no sparsity pattern in $\vec{\Omega}$, band structure
or factor model, but cannot use structure tied to the coordinates. In Section~\ref{sec:data}, every
estimator with nominally lower portfolio variance than hBayes used such structure, although none was
significantly better. As the results of Appendix~\ref{app:equivariance} hold for any compact group, in
future work we will consider equivariance to subgroups of $\opO(p)$ that retain such structure. In portfolio
analysis, the natural subgroup fixes the direction $\mathbf{1}$ of the budget constraint. For the graph
Laplacian, it is the group of permutation matrices $P$, which relabel the nodes,
$\mathcal{L} \mapsto P^\top \mathcal{L} P$, and the Metropolis steps on $\opO(p)$ would be replaced by cheap
transpositions of labels.

\section*{AI disclosure}
The authors used Claude (Anthropic) as an assistant during this work: to help draft and edit parts of
the text, to write and debug code for the simulation study and the data example, to check and simplify
derivations, and to search the literature. The research questions, the hierarchical Bayes framework and
the main theoretical results are the authors' own. The authors directed this assistance, checked the
derivations, code, references and numerical results, and revised the text. They take full responsibility
for the content of the paper.
\section*{Acknowledgements}
The authors are grateful to Jake Soloff for his contributions to deriving the Bayes rules,
to Peter Hoff for helpful discussions and for providing his lecture notes,
and to Alessandro Fulci for help in running the simulation study.

\bibliographystyle{plainnat}
\small
\bibliography{references}

\begin{thebibliography}{34}
\providecommand{\natexlab}[1]{#1}
\providecommand{\url}[1]{\texttt{#1}}
\expandafter\ifx\csname urlstyle\endcsname\relax
  \providecommand{\doi}[1]{doi: #1}\else
  \providecommand{\doi}{doi: \begingroup \urlstyle{rm}\Url}\fi

\bibitem[Berger(1985)]{Berger1985}
James~O. Berger.
\newblock \emph{Statistical Decision Theory and {Bayesian} Analysis}.
\newblock Springer Series in Statistics. Springer-Verlag, New York, 2nd
  edition, 1985.

\bibitem[Bickel and Levina(2008)]{bickel2008regularized}
Peter~J. Bickel and Elizaveta Levina.
\newblock Regularized estimation of large covariance matrices.
\newblock \emph{The Annals of Statistics}, 36\penalty0 (1):\penalty0 199--227,
  2008.

\bibitem[Chan et~al.(1999)Chan, Karceski, and Lakonishok]{Chan1999}
Louis K.~C. Chan, Jason Karceski, and Josef Lakonishok.
\newblock On portfolio optimization: forecasting covariances and choosing the
  risk model.
\newblock \emph{The Review of Financial Studies}, 12\penalty0 (5):\penalty0
  937--974, 1999.

\bibitem[DeMiguel et~al.(2009)DeMiguel, Garlappi, and Uppal]{DeMiguel2009}
Victor DeMiguel, Lorenzo Garlappi, and Raman Uppal.
\newblock Optimal versus naive diversification: how inefficient is the 1/{N}
  portfolio strategy?
\newblock \emph{The Review of Financial Studies}, 22\penalty0 (5):\penalty0
  1915--1953, 2009.

\bibitem[Diaconis and Shahshahani(1987)]{DiaconisShahshahani1987}
Persi Diaconis and Mehrdad Shahshahani.
\newblock The subgroup algorithm for generating uniform random variables.
\newblock \emph{Probability in the Engineering and Informational Sciences},
  1:\penalty0 15--32, 1987.

\bibitem[Donoho et~al.(2018)Donoho, Gavish, and Johnstone]{donoho2018optimal}
David~L. Donoho, Matan Gavish, and Iain~M. Johnstone.
\newblock Optimal shrinkage of eigenvalues in the spiked covariance model.
\newblock \emph{The Annals of Statistics}, 46\penalty0 (4):\penalty0
  1742--1778, 2018.

\bibitem[Eaton(1989)]{Eaton1989}
Morris~L. Eaton.
\newblock \emph{Group Invariance Applications in Statistics}, volume~1 of
  \emph{Regional Conference Series in Probability and Statistics}.
\newblock Institute of Mathematical Statistics and American Statistical
  Association, Hayward, CA and Alexandria, VA, 1989.

\bibitem[Fan et~al.(2013)Fan, Liao, and Mincheva]{Fan2013}
Jianqing Fan, Yuan Liao, and Martina Mincheva.
\newblock Large covariance estimation by thresholding principal orthogonal
  complements.
\newblock \emph{Journal of the Royal Statistical Society: Series B (Statistical
  Methodology)}, 75\penalty0 (4):\penalty0 603--680, 2013.

\bibitem[French(2026)]{French2026}
Kenneth~R. French.
\newblock Data library: 49 industry portfolios.
\newblock
  \url{https://mba.tuck.dartmouth.edu/pages/faculty/ken.french/data_library.html},
  2026.
\newblock Monthly value-weighted returns, based on the 202608 CRSP database.
  Accessed October 2026.

\bibitem[Friedman et~al.(2008)Friedman, Hastie, and
  Tibshirani]{friedman2008sparse}
Jerome Friedman, Trevor Hastie, and Robert Tibshirani.
\newblock Sparse inverse covariance estimation with the graphical lasso.
\newblock \emph{Biostatistics}, 9\penalty0 (3):\penalty0 432--441, 2008.

\bibitem[Friedman et~al.(2019)Friedman, Hastie, and Tibshirani]{glasso_man}
Jerome Friedman, Trevor Hastie, and Rob Tibshirani.
\newblock \emph{glasso: Graphical Lasso: Estimation of Gaussian Graphical
  Models}, 2019.
\newblock URL \url{https://CRAN.R-project.org/package=glasso}.
\newblock R package version 1.11.

\bibitem[Haff(1980)]{Haff1980}
L.~R. Haff.
\newblock Empirical {Bayes} estimation of the multivariate normal covariance
  matrix.
\newblock \emph{The Annals of Statistics}, 8\penalty0 (3):\penalty0 586--597,
  1980.

\bibitem[Haff(1991)]{Haff1991}
L.~R. Haff.
\newblock The variational form of certain {Bayes} estimators.
\newblock \emph{The Annals of Statistics}, 19\penalty0 (3):\penalty0
  1163--1190, 1991.

\bibitem[Hoff(2009)]{hoff2009simulation}
Peter~D. Hoff.
\newblock Simulation of the matrix {Bingham--von Mises--Fisher} distribution,
  with applications to multivariate and relational data.
\newblock \emph{Journal of Computational and Graphical Statistics}, 18\penalty0
  (2):\penalty0 438--456, 2009.

\bibitem[Jagannathan and Ma(2003)]{JagannathanMa2003}
Ravi Jagannathan and Tongshu Ma.
\newblock Risk reduction in large portfolios: why imposing the wrong
  constraints helps.
\newblock \emph{The Journal of Finance}, 58\penalty0 (4):\penalty0 1651--1683,
  2003.

\bibitem[James and Stein(1961)]{JamesStein1961}
William James and Charles Stein.
\newblock Estimation with quadratic loss.
\newblock In \emph{Proceedings of the Fourth {Berkeley} Symposium on
  Mathematical Statistics and Probability}, volume~1, pages 361--379, Berkeley,
  CA, 1961. University of California Press.

\bibitem[Jupp and Mardia(1979)]{jupp1979maximum}
Peter~E. Jupp and Kanti~V. Mardia.
\newblock Maximum likelihood estimators for the matrix von {Mises--Fisher} and
  {Bingham} distributions.
\newblock \emph{The Annals of Statistics}, 7\penalty0 (3):\penalty0 599--606,
  1979.

\bibitem[Kent et~al.(2018)Kent, Ganeiber, and Mardia]{kent2018new}
John~T. Kent, Asaad~M. Ganeiber, and Kanti~V. Mardia.
\newblock A new unified approach for the simulation of a wide class of
  directional distributions.
\newblock \emph{Journal of Computational and Graphical Statistics}, 27\penalty0
  (2):\penalty0 291--301, 2018.

\bibitem[Ledoit and Wolf(2003)]{LedoitWolf2003}
Olivier Ledoit and Michael Wolf.
\newblock Improved estimation of the covariance matrix of stock returns with an
  application to portfolio selection.
\newblock \emph{Journal of Empirical Finance}, 10:\penalty0 603--621, 2003.

\bibitem[Ledoit and Wolf(2004{\natexlab{a}})]{LedoitWolf2004}
Olivier Ledoit and Michael Wolf.
\newblock A well-conditioned estimator for large-dimensional covariance
  matrices.
\newblock \emph{Journal of Multivariate Analysis}, 88\penalty0 (2):\penalty0
  365--411, 2004{\natexlab{a}}.

\bibitem[Ledoit and Wolf(2004{\natexlab{b}})]{LedoitWolf2004jpm}
Olivier Ledoit and Michael Wolf.
\newblock Honey, {I} shrunk the sample covariance matrix.
\newblock \emph{Journal of Portfolio Management}, 30\penalty0 (4):\penalty0
  110--119, 2004{\natexlab{b}}.

\bibitem[Ledoit and Wolf(2011)]{LedoitWolf2011}
Olivier Ledoit and Michael Wolf.
\newblock Robust performance hypothesis testing with the variance.
\newblock \emph{Wilmott Magazine}, 55:\penalty0 86--89, September 2011.

\bibitem[Ledoit and Wolf(2012)]{LedoitWolf2012}
Olivier Ledoit and Michael Wolf.
\newblock Nonlinear shrinkage estimation of large-dimensional covariance
  matrices.
\newblock \emph{The Annals of Statistics}, 40\penalty0 (2):\penalty0
  1024--1060, 2012.

\bibitem[Ledoit and Wolf(2015)]{LedoitWolf2015}
Olivier Ledoit and Michael Wolf.
\newblock Spectrum estimation: A unified framework for covariance matrix
  estimation and {PCA} in large dimensions.
\newblock \emph{Journal of Multivariate Analysis}, 139:\penalty0 360--384,
  2015.

\bibitem[Ledoit and Wolf(2017)]{LedoitWolf2017}
Olivier Ledoit and Michael Wolf.
\newblock Nonlinear shrinkage of the covariance matrix for portfolio selection:
  {M}arkowitz meets {G}oldilocks.
\newblock \emph{The Review of Financial Studies}, 30\penalty0 (12):\penalty0
  4349--4388, 2017.

\bibitem[Ledoit and Wolf(2018)]{LedoitWolf2018}
Olivier Ledoit and Michael Wolf.
\newblock Optimal estimation of a large-dimensional covariance matrix under
  {Stein's} loss.
\newblock \emph{Bernoulli}, 24\penalty0 (4B):\penalty0 3791--3832, 2018.
\newblock \doi{10.3150/17-BEJ979}.

\bibitem[Ledoit and Wolf(2020)]{LedoitWolf2020}
Olivier Ledoit and Michael Wolf.
\newblock Analytical nonlinear shrinkage of large-dimensional covariance
  matrices.
\newblock \emph{The Annals of Statistics}, 48\penalty0 (5):\penalty0
  3043--3065, 2020.

\bibitem[Lehmann and Casella(1998)]{LehmannCasella1998}
Erich~L. Lehmann and George Casella.
\newblock \emph{Theory of Point Estimation}.
\newblock Springer, New York, 2nd edition, 1998.

\bibitem[Lin and Perlman(1985)]{LinPerlman1985}
S.~P. Lin and Michael~D. Perlman.
\newblock A {Monte Carlo} comparison of four estimators of a covariance matrix.
\newblock In P.~R. Krishnaiah, editor, \emph{Multivariate Analysis VI}, pages
  411--429. North-Holland, Amsterdam, 1985.

\bibitem[Naul(2016)]{stcov_man}
Brett Naul.
\newblock \emph{stcov: Stein's Covariance Estimator}, 2016.
\newblock URL \url{https://CRAN.R-project.org/package=stcov}.
\newblock R package version 0.1.0.

\bibitem[Newey and West(1987)]{NeweyWest1987}
Whitney~K. Newey and Kenneth~D. West.
\newblock A simple, positive semi-definite, heteroskedasticity and
  autocorrelation consistent covariance matrix.
\newblock \emph{Econometrica}, 55\penalty0 (3):\penalty0 703--708, 1987.

\bibitem[Ramprasad(2016)]{nlshrink_man}
Anirudh Ramprasad.
\newblock \emph{nlshrink: Non-Linear Shrinkage Estimation of Population
  Eigenvalues and Covariance Matrices}, 2016.
\newblock URL \url{https://CRAN.R-project.org/package=nlshrink}.
\newblock R package version 1.0.1.

\bibitem[Stein(1975)]{Stein1975}
Charles Stein.
\newblock Estimation of a covariance matrix.
\newblock Rietz Lecture, 39th Annual Meeting of the IMS, Atlanta, GA, 1975.

\bibitem[Weinstein et~al.(2025)Weinstein, Wallin, Yekutieli, and
  Bogdan]{WeinsteinEtAl2025}
Asaf Weinstein, Jonas Wallin, Daniel Yekutieli, and Malgorzata Bogdan.
\newblock Nonparametric shrinkage estimation in high dimensional {GLMs} via
  {P}olya trees.
\newblock \emph{Statistica Sinica}, 2025.
\newblock \doi{10.5705/ss.202025.0221}.
\newblock In press.

\end{thebibliography}

\newpage

\appendix
\counterwithin{thm}{section}

\section{Equivariance and Haar measure Bayes rules}
\label{app:equivariance}

The results in this appendix are classical; see \citet[Chapter~6]{Berger1985},
\citet{Eaton1989} and \citet[Chapter~3]{LehmannCasella1998} for equivariance
and invariant decision problems, and \citet{JamesStein1961} for the
covariance estimation problem.
Throughout, we use Definition~\ref{def:equivariance} of Section~\ref{subsec:equivariance-GL} with a
general group ${\cal G}$ and parameter $\theta \in \Theta$. The risk of an estimator $\hat\theta$ under
the loss $L$ is $R(L, \hat\theta; \theta) = E_{X \sim f(x \mid \theta)}\, L(\hat\theta(X), \theta)$.

\begin{thm} \label{cor1}
If the family of densities and the estimator are equivariant and the loss
function is invariant under ${\cal G}$, then the risk of the estimator is constant on
each orbit $\{\bar g(\theta) : g \in {\cal G} \}$, $\theta \in \Theta$.
\end{thm}

\begin{proof}
Let $\theta_0 \in \Theta$ and $g \in {\cal G}$, and let $X \sim f(x \mid \theta_0)$, so
that $Y = g(X) \sim f(y \mid \bar g(\theta_0))$. By the invariance of the loss and
the equivariance of the estimator,
\[
  L(\hat\theta(X), \theta_0)
  = L(\bar g(\hat\theta(X)), \bar g(\theta_0))
  = L(\hat\theta(Y), \bar g(\theta_0)).
\]
Taking expectations,
\[
  R(L, \hat\theta; \theta_0)
  = E_{X \sim f(x \mid \theta_0)} L(\hat\theta(X), \theta_0)
  = E_{Y \sim f(y \mid \bar g(\theta_0))} L(\hat\theta(Y), \bar g(\theta_0))
  = R(L, \hat\theta; \bar g(\theta_0)).
\]
\end{proof}

\noindent In particular, if ${\cal G}$ is transitive on $\Theta$, the risk is constant on $\Theta$.

\begin{lemma} \label{thm221}
If the family of densities is equivariant under ${\cal G}$ and the prior
distribution of $\theta$ is invariant under ${\cal G}$,
then the posterior distribution of $\theta$ is equivariant under ${\cal G}$: if
$\theta \sim \pi(\cdot \mid x)$ then $\bar g(\theta) \sim \pi(\cdot \mid g(x))$.
\end{lemma}

\begin{proof}
Take the prior density $\pi$ with respect to a ${\cal G}$-invariant measure on
$\Theta$, so that $\pi(\bar g(\theta)) = \pi(\theta)$. By equivariance of the
family, $f(g(x) \mid \bar g(\theta)) = c_g(x) f(x \mid \theta)$, where $c_g(x)$
does not depend on $\theta$; for the orthogonal transformations considered in
this paper, $c_g(x) = 1$. Hence
\[
  \pi(\bar g(\theta) \mid g(x))
  \propto f(g(x) \mid \bar g(\theta))\, \pi(\bar g(\theta))
  = c_g(x)\, f(x \mid \theta)\, \pi(\theta)
  \propto \pi(\theta \mid x).
\]
\end{proof}

We now assume that the family of densities is equivariant under ${\cal G}$ and
that ${\cal G}$ is compact, so that its Haar measure can be normalized to a
probability distribution, $\pi^{\opHaar}$. We specify a generative model for
$(\theta, X)$: for a fixed parameter value $\theta_0$, prior distribution samples
are given by $\bar{g}(\theta_0)$ for $g \sim \pi^{\opHaar}$, and the data samples
are $X \sim f(x \mid \bar{g}(\theta_0))$. The prior is supported on the orbit
$\{\bar g(\theta_0) : g \in {\cal G}\}$ and is invariant under ${\cal G}$. The
Haar measure Bayes rule is the estimator that minimizes the Haar measure
average risk,
\begin{equation}
  E_{g \sim \pi^{\opHaar}} \, R(L, \hat{\theta}; \bar{g}(\theta_0)). \label{eq33}
\end{equation}
It is obtained by minimizing, for each $x$, the posterior expected loss
$\rho(a \mid x) = \mathrm E[L(a, \theta) \mid x]$ over the set ${\cal A}$ of values of the estimators
considered, which we assume is mapped onto itself by each $\bar g$.
We assume that this minimizer exists and is unique. For the losses $L_0$, $L_1$ and $L_2$, with
${\cal A}$ the positive semidefinite matrices (see the paragraph after Definition~\ref{oracle-lambda-BR-def}),
this holds because the posterior expected loss is strictly convex and tends to infinity at the boundary
of ${\cal A}$ under $L_1$ and as the estimate grows without bound.

\begin{thm} \label{cor11}

If the loss is invariant under ${\cal G}$, then the Haar measure Bayes rule
$\hat\theta^{\opHaar}$ is equivariant under ${\cal G}$, and it is the unique
equivariant estimator with the smallest risk at every parameter value in the
orbit $\{\bar g(\theta_0) : g \in {\cal G}\}$.
\end{thm}

\begin{proof}
By Lemma~\ref{thm221}, if $\theta \sim \pi(\cdot \mid x)$ then
$\bar g(\theta) \sim \pi(\cdot \mid g(x))$. With the invariance of the loss,
this gives, for every $a \in {\cal A}$,
\[
  \rho(\bar g(a) \mid g(x))
  = \mathrm E[L(\bar g(a), \bar g(\theta)) \mid x]
  = \mathrm E[L(a, \theta) \mid x]
  = \rho(a \mid x).
\]
Since $\bar g$ is a bijection of ${\cal A}$, it follows that
$\min_a \rho(a \mid g(x)) = \min_a \rho(a \mid x)$, and that
$\bar g(\hat\theta^{\opHaar}(x))$ attains this minimum at $g(x)$. By uniqueness,
$\bar g(\hat\theta^{\opHaar}(x)) = \hat\theta^{\opHaar}(g(x))$, that is,
$\hat\theta^{\opHaar}$ is equivariant.

By Theorem~\ref{cor1}, every equivariant estimator has constant risk on the orbit of
$\theta_0$, which therefore equals its Haar measure average risk
(\ref{eq33}). Since $\hat\theta^{\opHaar}$ minimizes (\ref{eq33}) over all
estimators with values in ${\cal A}$, its risk at every parameter value in the orbit is no larger than
that of any equivariant estimator. An equivariant estimator with the same risk
would have the same Haar measure average risk, and would therefore also be a
Haar measure Bayes rule; by uniqueness it equals $\hat\theta^{\opHaar}$.
\end{proof}

\section{Oracle Bayes rules for precision matrix estimation}
\label{app:precision}

For estimating the precision matrix we use the losses of Section~\ref{subsec:equivariance-GL} with
$(\hat{\boldsymbol\Sigma}, \boldsymbol\Sigma)$ replaced by
$(\hat{\boldsymbol\Omega}, \boldsymbol\Omega)$. Since
$\boldsymbol\Omega^{-1} = \boldsymbol\Sigma$, the posterior expected loss of
$\hat{\boldsymbol\Omega}$ has the same form as that of $\hat{\boldsymbol\Sigma}$ in
the proof of Proposition~\ref{prop:oracle-bayes-rules}, with the roles of $\tilde{\boldsymbol\Sigma}$ and
$\tilde{\boldsymbol\Omega}$ interchanged, and the derivation in Appendix~\ref{sec:proofs} applies
verbatim. By the same diagonalization argument, the oracle Bayes rule has the form
$\hat{\boldsymbol\Omega} = \mathbf V\operatorname{diag}(\boldsymbol\omega)\mathbf V^\top$
with $\boldsymbol\omega = \mathbf G\boldsymbol\delta$, where $\mathbf G$ and
$\mathbf N = \mathbf G^\top\mathbf G$ are as in Section~\ref{subsec:oracle-rules}; for $n < p$ this sets
$\omega_{n+1} = \cdots = \omega_p = \bar\omega$. Equivalently, each rule is obtained
from the corresponding rule in Proposition~\ref{prop:oracle-bayes-rules} by replacing $\lambda_j$ with
$\lambda_j^{-1}$.

\subsection{Frobenius loss}
\[
  \omega_k = \mathrm E[\tilde\Omega_{kk} \mid \mathbf x]
          = \sum_{j=1}^p \lambda_j^{-1}\,\mathrm E[\tilde\Gamma_{kj}^2 \mid \mathbf x],
  \qquad k \le \min(n,p),
\]
and, for $n < p$,
\[
  \bar\omega = \frac{1}{p-n}\sum_{k=n+1}^p \mathrm E[\tilde\Omega_{kk} \mid \mathbf x].
\]

\subsection{Stein loss}
\[
  \omega_k = \big(\mathrm E[\tilde\Sigma_{kk} \mid \mathbf x]\big)^{-1}
          = \Big(\sum_{j=1}^p \lambda_j\,\mathrm E[\tilde\Gamma_{kj}^2 \mid \mathbf x]\Big)^{-1},
  \qquad k \le \min(n,p),
\]
and, for $n < p$,
\[
  \bar\omega = \Big(\frac{1}{p-n}\sum_{k=n+1}^p \mathrm E[\tilde\Sigma_{kk} \mid \mathbf x]\Big)^{-1}.
\]
In terms of the covariance matrix, $\hat{\boldsymbol\Omega}^{-1}$ under this loss is the
estimator of $\boldsymbol\Sigma$ that is optimal under the inverse Stein's loss of
\citet{LedoitWolf2018}.

\subsection{Squared Stein loss}
Provided its solution has positive entries (otherwise see the paragraph after
Proposition~\ref{prop:oracle-bayes-rules}), the vector $\boldsymbol\delta$ solves
\[
  \mathbf G^\top\mathbf A\,\mathbf G\,\boldsymbol\delta = \mathbf G^\top\mathbf b,
  \qquad
  \mathbf A = \mathrm E[\tilde{\boldsymbol\Sigma}\circ\tilde{\boldsymbol\Sigma} \mid \mathbf x],
  \quad
  \mathbf b = \operatorname{diag}\big(\mathrm E[\tilde{\boldsymbol\Sigma} \mid \mathbf x]\big),
\]
where
\[
  A_{k\ell} = \sum_{i,j=1}^p \lambda_i\lambda_j\,
    \mathrm E[\tilde\Gamma_{ki}\tilde\Gamma_{kj}\tilde\Gamma_{\ell i}\tilde\Gamma_{\ell j} \mid \mathbf x],
  \qquad
  b_k = \sum_{j=1}^p \lambda_j\,\mathrm E[\tilde\Gamma_{kj}^2 \mid \mathbf x].
\]
For $n \ge p$ this reduces to $\mathbf A\boldsymbol\omega = \mathbf b$. For $n < p$ it is
an $(n+1)\times(n+1)$ system in $(\omega_1, \dots, \omega_n, \bar\omega)$, whose last row
and column sum the entries of $\mathbf A$ and $\mathbf b$ over the indices $n+1, \dots, p$.

\section{Proofs}\label{sec:proofs}

 \begin{proof}[Proof of Proposition \ref{prop12}]
By equivariance, $\hat{\vec{\Sigma}}(\vec{S}) = \hat{\vec{\Sigma}}(\vec{V}\vec{L}\vec{V}^\top) = \vec{V}\hat{\vec{\Sigma}}(\vec{L})\vec{V}^\top$
(take $\vec{R} = \vec{V}^\top$ in the definition of Section~\ref{subsec:equivariance-O}), so it suffices to
show that $\hat{\vec{\Sigma}}(\vec{L})$ is diagonal.
For $i = 1, \cdots, p$, let $\vec{J}_i$ be the diagonal matrix with $1$ along the diagonal except for $-1$ at the $i$'th spot.
As $\vec{L}$ is diagonal then  $\vec{L}  = \vec{J}_i^\top  \vec{L} \vec{J}_i$.
As $\vec{J}_i$ is orthogonal then equivariance to multiplication by orthogonal matrices yields,
\[
\hat{\vec{\Sigma}} (\vec{L} )  =
\hat{\vec{\Sigma}} (\vec{J}_i^\top  \vec{L} \vec{J}_i)  =
\vec{J}_i^\top \hat{\vec{\Sigma}} (\vec{L} ) \vec{J}_i.
\]
This implies that the off-diagonal entries of $\hat{\vec{\Sigma}} (\vec{L})$ are $0$.

We will now assume that $n \le p-2$.
For $n < i < j \le p$, let $\vec{B}_{i,j}$ denote the orthogonal matrix  constructed by taking the $p \times p$ identity matrix and swapping Column $i$ with Column $j$.
As $l_i = 0$ for $i = n+1, \cdots, p$ then
$\vec{L}  = \vec{B}_{i,j}^\top \vec{L} \vec{B}_{i,j}$.
Invoking  equivariance to multiplication by orthogonal matrices yields,
\[
\hat{\vec{\Sigma}} (\vec{L} )  =
\vec{B}_{i,j}^\top \hat{\vec{\Sigma}} (\vec{L} ) \vec{B}_{i,j}.
\]
This implies that $d_i = d_j$, and in general $d_{n+1} = \cdots = d_p$.
In particular for the no data case, corresponding to $n = 0$, the covariance matrix estimator is  fixed with
$\hat{\vec{\Sigma}}  \equiv  \opdiag( d_1, \cdots, d_1)$.
\end{proof}

\bigskip
\begin{proof} [Proof of Lemma \ref{lemm2a}]
Property (a), which we express $\bar{g}_{\vec{R} }( \vec{\Omega})
= \bar{g}_{\vec{R} }( \vec{\Sigma})^{-1}$, follows from
\[
\bar{g}_{\vec{R} }( \vec{\Omega})  \bar{g}_{\vec{R} }( \vec{\Sigma})
= ( \vec{R}^\top  \vec{\Omega} \vec{R}  ) ( \vec{R}^\top  \vec{\Sigma} \vec{R} )
=  \vec{R}^\top  \vec{\Omega}  (\vec{R}  \vec{R}^\top ) \vec{\Sigma} \vec{R}
= \vec{I}.
\]

\noindent Property (b) follows from equivariance of the sample covariance matrix and
Property (a) applied to the nonsingular precision matrix of
$g_{\vec{R}} (\vec{X})$ for $n \ge p$,
\[
( \vec{S} ( g_{\vec{R}} (\vec{X})))^{-1} =
(\bar{g}_{\vec{R}} ( \vec{S} (\vec{X})))^{-1} =
\bar{g}_{\vec{R}} ( ( \vec{S} (\vec{X}))^{-1}).
\]
\end{proof}

\bigskip
\begin{proof} [Proof of Lemma \ref{lemm4a}]
$\vec{\Lambda} = \vec{B} \vec{\Lambda} \vec{B}^\top$ implies that $\vec{\Lambda}^{-1} = \vec{B} \vec{\Lambda}^{-1} \vec{B}^\top$, yielding
\begin{eqnarray}
\pi^{\opHaar} ( \vec{\Gamma} \vec{B}  | \vec{x}, \vec{\lambda}) & \propto & \exp ( \optr( -  \vec{\Lambda}^{-1} \vec{B}^\top \vec{\Gamma}^\top  ( n \vec{S} ) \vec{\Gamma} \vec{B} / 2)) \nonumber \\
& = &\exp ( \optr( -   \vec{B} \vec{\Lambda}^{-1} \vec{B}^\top \vec{\Gamma}^\top  ( n \vec{S} ) \vec{\Gamma}  / 2))\label{cyc-trace} \\
& = &\exp ( \optr( -   \vec{\Lambda}^{-1}  \vec{\Gamma}^\top  ( n \vec{S} ) \vec{\Gamma}  / 2)) \propto \pi^{\opHaar} ( \vec{\Gamma} | \vec{x}, \vec{\lambda}), \nonumber
\end{eqnarray}
where the equality in (\ref{cyc-trace}) is due to the invariance of the trace to circular shifts.
\end{proof}

\bigskip
\begin{proof} [Proof of Lemma \ref{lemm5a}]
To prove the lemma we express,
\begin{eqnarray*}
\lefteqn{\pi^{\opHaar} \left(  \vec{\Gamma}  ( \vec{\Gamma}^\top \vec{V}\vec{B}^\top \vec{V}^\top \vec{\Gamma} )  | \vec{x}, \vec{\lambda}\right)
 = \pi^{\opHaar} \left( \vec{V}\vec{B}^\top \vec{V}^\top \vec{\Gamma}  | \vec{x}, \vec{\lambda}\right) } \\
& \propto & \exp \left[ \optr\left\{ -  \vec{\Lambda}^{-1}  \left( \vec{V}\vec{B}^\top \vec{V}^\top \vec{\Gamma}\right)^\top \vec{S}
 \vec{V}\vec{B}^\top \vec{V}^\top \vec{\Gamma} ( n / 2) \right\} \right] \nonumber \\
& = & \exp \left[ \optr\left\{ -  \vec{\Lambda}^{-1} \vec{\Gamma}^\top \vec{V} \vec{B} \vec{V}^\top ( \vec{V} \vec{L}   \vec{V}^\top )
 \vec{V}\vec{B}^\top \vec{V}^\top \vec{\Gamma} ( n / 2) \right\} \right] \nonumber \\
& = & \exp \left[ \optr\left\{ -  \vec{\Lambda}^{-1} \vec{\Gamma}^\top \vec{V} \vec{B}  \vec{L}
\vec{B}^\top \vec{V}^\top \vec{\Gamma} ( n / 2) \right\} \right] \nonumber \\
& = & \exp \left[ \optr\left\{ -  \vec{\Lambda}^{-1} \vec{\Gamma}^\top \vec{V}   \vec{L}  \vec{V}^\top \vec{\Gamma} ( n / 2) \right\} \right] \nonumber \\
& = & \exp \left[ \optr\left\{ -  \vec{\Lambda}^{-1}   \vec{\Gamma}^\top   \vec{S} \vec{\Gamma}  ( n / 2) \right\} \right] \propto \pi^{\opHaar} ( \vec{\Gamma} | \vec{x}, \vec{\lambda})
\end{eqnarray*}
\end{proof}

\bigskip
\begin{proof} [Proof of Lemma \ref{cor34}]
For $\vec{R} \in \opO(p)$, $\bar{g}_{\vec{R}} ( \vec{\Sigma}) = ( \vec{R}^\top \vec{\Gamma}) \vec{\Lambda} ( \vec{\Gamma}^\top \vec{R})$
reveals that multiplication by $\vec{R}$ maps eigenvector matrix $\vec{\Gamma}$ to
$\vec{R}^\top \vec{\Gamma}\in \opO(p)$.
Equivariance of the sample covariance matrix, $\vec{S} ( g_{\vec{R}} (\vec{x})) =
\bar{g}_{\vec{R}} ( \vec{S} ( \vec{x}))$, yields,
\begin{eqnarray*}
\lefteqn{\pi^{\opHaar} ( \vec{R}^\top \vec{\Gamma} | g_{\vec{R}} (\vec{x}), \vec{\lambda})
\propto  \exp \left[ \optr \left\{ -  \vec{\Lambda}^{-1}
( \vec{R}^\top \vec{\Gamma})^\top   \vec{S} ( g_{\vec{R}} (\vec{x}))
  \vec{R}^\top \vec{\Gamma} (n /  2) \right\} \right] } \\
&   =  &  \exp \left[ \optr \left\{ -  \vec{\Lambda}^{-1} \vec{\Gamma}^\top
\vec{R} \vec{R}^\top  \vec{S} (\vec{x}) \vec{R}  \vec{R}^\top \vec{\Gamma} (n/ 2) \right\} \right]
\propto \pi^{\opHaar} ( \vec{\Gamma} | \vec{x}, \vec{\lambda}).
\end{eqnarray*}
\end{proof}

\bigskip

\begin{proof}[Proof of Proposition~\ref{prop:oracle-bayes-rules}]
By Lemma~\ref{prop2w} and Proposition~\ref{prop12}, the oracle Bayes rule has the form
$\hat{\boldsymbol\Sigma} = \mathbf V\mathbf D\mathbf V^\top$ with
$\mathbf D = \operatorname{diag}(\mathbf d)$ and $\mathbf d = \mathbf G\boldsymbol\delta$, 
so it suffices to minimize the posterior expected loss over $\boldsymbol\delta \ge \vec{0}$ (with $L_1 = \infty$ if some $\delta_i = 0$).
Each of the three losses is invariant under $\opO(p)$, hence
$L(\hat{\boldsymbol\Sigma}, \boldsymbol\Sigma)
  = L(\mathbf V^\top\hat{\boldsymbol\Sigma}\mathbf V, \mathbf V^\top\boldsymbol\Sigma\mathbf V)
  = L(\mathbf D, \tilde{\boldsymbol\Sigma})$,
and $(\mathbf V^\top\boldsymbol\Sigma\mathbf V)^{-1} = \tilde{\boldsymbol\Omega}$.
Since $\mathbf V$ is a function of $\mathbf x$, posterior expectations of
functions of $\tilde{\boldsymbol\Sigma}$ and $\tilde{\boldsymbol\Omega}$ are
well defined given $\mathbf x$. Throughout, ``const'' denotes terms that do not
depend on $\boldsymbol\delta$.

\emph{Frobenius loss.} Expanding the square,
\[
  L_0(\mathbf D, \tilde{\boldsymbol\Sigma})
  = \sum_{k=1}^p d_k^2 - 2\sum_{k=1}^p d_k\tilde\Sigma_{kk}
    + \|\tilde{\boldsymbol\Sigma}\|_F^2 .
\]
Substituting $\mathbf d = \mathbf G\boldsymbol\delta$ and taking posterior expectations,
\[
  \mathrm E[L_0 \mid \mathbf x]
  = \boldsymbol\delta^\top\mathbf N\boldsymbol\delta
    - 2\boldsymbol\delta^\top\mathbf G^\top
      \operatorname{diag}\big(\mathrm E[\tilde{\boldsymbol\Sigma} \mid \mathbf x]\big)
    + \text{const}.
\]
Since $\mathbf N$ is diagonal with positive entries, this is a strictly convex
quadratic in $\boldsymbol\delta$, minimized at
$\boldsymbol\delta = \mathbf N^{-1}\mathbf G^\top
  \operatorname{diag}(\mathrm E[\tilde{\boldsymbol\Sigma} \mid \mathbf x])$.
That is, $d_k = \mathrm E[\tilde\Sigma_{kk} \mid \mathbf x]$ for $k \le \min(n,p)$
and, for $n<p$, the tied entry is the average of
$\mathrm E[\tilde\Sigma_{kk} \mid \mathbf x]$ over $k > n$. Its entries are positive, so it is also the
minimizer over $\boldsymbol\delta \ge \vec{0}$.

\emph{Stein loss.} Since $\mathbf D$ is diagonal,
$\operatorname{tr}(\mathbf D\tilde{\boldsymbol\Omega}) = \sum_k d_k\tilde\Omega_{kk}$ and
$\log\det(\mathbf D\tilde{\boldsymbol\Omega})
  = \sum_k \log d_k + \log\det\tilde{\boldsymbol\Omega}$. Hence
\[
  \mathrm E[L_1 \mid \mathbf x]
  = \sum_{i=1}^m \big(\delta_i c_i - N_{ii}\log\delta_i\big) + \text{const},
  \qquad
  \mathbf c = \mathbf G^\top
    \operatorname{diag}\big(\mathrm E[\tilde{\boldsymbol\Omega} \mid \mathbf x]\big).
\]
The objective is separable, and each term is strictly convex in $\delta_i > 0$
with minimizer $\delta_i = N_{ii}/c_i$. That is,
$d_k = 1/\mathrm E[\tilde\Omega_{kk} \mid \mathbf x]$ for $k \le \min(n,p)$ and,
for $n<p$, the tied entry is the reciprocal of the average of
$\mathrm E[\tilde\Omega_{kk} \mid \mathbf x]$ over $k > n$.

\emph{Squared Stein loss.} Expanding the square,
\[
  L_2(\mathbf D, \tilde{\boldsymbol\Sigma})
  = \operatorname{tr}\big((\mathbf D\tilde{\boldsymbol\Omega})^2\big)
    - 2\operatorname{tr}(\mathbf D\tilde{\boldsymbol\Omega}) + p .
\]
Because $\mathbf D$ is diagonal and $\tilde{\boldsymbol\Omega}$ is symmetric,
\[
  \operatorname{tr}\big((\mathbf D\tilde{\boldsymbol\Omega})^2\big)
  = \sum_{k,\ell=1}^p d_k\,\tilde\Omega_{k\ell}\,d_\ell\,\tilde\Omega_{\ell k}
  = \mathbf d^\top(\tilde{\boldsymbol\Omega}\circ\tilde{\boldsymbol\Omega})\,\mathbf d,
  \qquad
  \operatorname{tr}(\mathbf D\tilde{\boldsymbol\Omega})
  = \mathbf d^\top\operatorname{diag}(\tilde{\boldsymbol\Omega}).
\]
Taking posterior expectations,
\[
  \mathrm E[L_2 \mid \mathbf x]
  = \boldsymbol\delta^\top\mathbf G^\top\mathbf A\,\mathbf G\boldsymbol\delta
    - 2\boldsymbol\delta^\top\mathbf G^\top\mathbf b + p,
\]
with $\mathbf A = \mathrm E[\tilde{\boldsymbol\Omega}\circ\tilde{\boldsymbol\Omega} \mid \mathbf x]$
and $\mathbf b = \operatorname{diag}(\mathrm E[\tilde{\boldsymbol\Omega} \mid \mathbf x])$.
Since $\tilde{\boldsymbol\Omega}$ is positive definite, the Schur product theorem gives that
$\tilde{\boldsymbol\Omega}\circ\tilde{\boldsymbol\Omega}$ is positive definite, and hence so is
$\mathbf A$. As $\mathbf G$ has full column rank, $\mathbf G^\top\mathbf A\mathbf G$ is also
positive definite. The posterior expected loss is therefore strictly convex in $\boldsymbol\delta$,
and its unique minimizer over $\mathbb R^m$ is the solution of
$\mathbf G^\top\mathbf A\mathbf G\boldsymbol\delta = \mathbf G^\top\mathbf b$.
Since $\hat{\boldsymbol\Sigma}$ is positive semidefinite if and only if $\boldsymbol\delta \ge \vec{0}$, the
oracle Bayes rule is given by the minimizer $\boldsymbol\delta^+$ of this strictly convex quadratic over
$\boldsymbol\delta \ge \vec{0}$, which exists and is unique. If the solution of
$\mathbf G^\top\mathbf A\mathbf G\boldsymbol\delta = \mathbf G^\top\mathbf b$ has positive entries, it is
$\boldsymbol\delta^+$. Otherwise the unconstrained minimizer is not in the open orthant, so
$\boldsymbol\delta^+$ lies on its boundary: with $F = \{i : \delta^+_i > 0\}$, $F \ne \{1, \dots, m\}$.
By the Karush--Kuhn--Tucker conditions,
$\delta^+_i = 0$ for $i \notin F$ and
\[
  (\mathbf G^\top\mathbf A\mathbf G)_{FF}\,\boldsymbol\delta^+_F = (\mathbf G^\top\mathbf b)_F,
  \qquad
  (\mathbf G^\top\mathbf A\mathbf G\boldsymbol\delta^+ - \mathbf G^\top\mathbf b)_i \ge 0
  \quad (i \notin F),
\]
so $\hat{\boldsymbol\Sigma}$ is singular; $\boldsymbol\delta^+$ can be computed by any convex quadratic
programming routine.

Finally, writing $\tilde{\boldsymbol\Sigma} = \sum_j \lambda_j\tilde{\boldsymbol\Gamma}_j\tilde{\boldsymbol\Gamma}_j^\top$
and $\tilde{\boldsymbol\Omega} = \sum_j \lambda_j^{-1}\tilde{\boldsymbol\Gamma}_j\tilde{\boldsymbol\Gamma}_j^\top$
gives
$\mathrm E[\tilde\Sigma_{kk} \mid \mathbf x] = \sum_j \lambda_j\mathrm E[\tilde\Gamma_{kj}^2 \mid \mathbf x]$,
$\mathrm E[\tilde\Omega_{kk} \mid \mathbf x] = \sum_j \lambda_j^{-1}\mathrm E[\tilde\Gamma_{kj}^2 \mid \mathbf x]$ and
$A_{k\ell} = \sum_{i,j}\lambda_i^{-1}\lambda_j^{-1}
  \mathrm E[\tilde\Gamma_{ki}\tilde\Gamma_{kj}\tilde\Gamma_{\ell i}\tilde\Gamma_{\ell j} \mid \mathbf x]$,
which are the moment expressions stated after Proposition~\ref{prop:oracle-bayes-rules}.
\end{proof}

\section{An exact posterior sampler in the single-spike model}\label{app:spike}
To assess the oracle Gibbs sampler, which applies to an arbitrary eigenvalue
vector, we compare it with an exact sampler available when
$\boldsymbol{\lambda} = (\lambda_1, \sigma^2, \dots, \sigma^2)$ with
$\lambda_1 > \sigma^2$.

\begin{proposition}\label{prop:spiked}
Let $\boldsymbol{\lambda} = (\lambda_1, \sigma^2, \dots, \sigma^2)$ with
$\lambda_1 > \sigma^2$, let $\mathbf{S} = \mathbf{V}\mathbf{L}\mathbf{V}^\top$,
and write $\boldsymbol{\Gamma} = \mathbf{V}\tilde{\boldsymbol{\Gamma}}$.
Under $\pi^{\mathrm{Haar}}(\boldsymbol{\Gamma} \mid \mathbf{x}, \boldsymbol{\lambda})$,
the first column $\mathbf{u} = \tilde{\boldsymbol{\Gamma}}_1$ follows the
Bingham distribution on $\mathbb{S}^{p-1}$,
\[
  \pi(\mathbf{u} \mid \mathbf{x}, \boldsymbol{\lambda})
  \propto \exp\!\left(c\, \mathbf{u}^\top \mathbf{L} \mathbf{u}\right),
  \qquad
  c = \frac{n}{2}\left(\frac{1}{\sigma^2} - \frac{1}{\lambda_1}\right) > 0,
\]
and $|\mathbf{V}_1^\top \boldsymbol{\Gamma}_1| = |u_1|$.
\end{proposition}

\begin{proof}
Since the Haar measure is invariant under left multiplication by $\mathbf{V}$,
$\tilde{\boldsymbol{\Gamma}} = \mathbf{V}^\top \boldsymbol{\Gamma}$ has posterior
density proportional to
$\exp\{-\tfrac{n}{2}\operatorname{tr}(\boldsymbol{\Lambda}^{-1}
\tilde{\boldsymbol{\Gamma}}^\top \mathbf{L} \tilde{\boldsymbol{\Gamma}})\}$
with respect to the Haar measure. Using
$\sum_{j=1}^p \tilde{\boldsymbol{\Gamma}}_j \tilde{\boldsymbol{\Gamma}}_j^\top = \mathbf{I}$,
\[
  \operatorname{tr}(\boldsymbol{\Lambda}^{-1}\tilde{\boldsymbol{\Gamma}}^\top
  \mathbf{L}\tilde{\boldsymbol{\Gamma}})
  = \sum_{j=1}^p \lambda_j^{-1}\tilde{\boldsymbol{\Gamma}}_j^\top \mathbf{L}
    \tilde{\boldsymbol{\Gamma}}_j
  = \left(\frac{1}{\lambda_1} - \frac{1}{\sigma^2}\right)
    \tilde{\boldsymbol{\Gamma}}_1^\top \mathbf{L}\tilde{\boldsymbol{\Gamma}}_1
    + \frac{1}{\sigma^2}\operatorname{tr}(\mathbf{L}).
\]
The posterior therefore depends on $\tilde{\boldsymbol{\Gamma}}$ only through
$\tilde{\boldsymbol{\Gamma}}_1$, which is uniformly distributed on
$\mathbb{S}^{p-1}$ under the Haar measure. Finally,
$|\mathbf{V}_1^\top \boldsymbol{\Gamma}_1|
= |\mathbf{V}_1^\top \mathbf{V}\tilde{\boldsymbol{\Gamma}}_1| = |u_1|$.
\end{proof}

Proposition~\ref{prop:spiked} implies that the posterior distribution of
$|\mathbf{V}_1^\top \boldsymbol{\Gamma}_1|$, and hence
$q^{\mathrm{Haar}}_\alpha$, depends on the data only through the sample
eigenvalues $\mathbf{l}$. Writing the density as
$\exp(-\mathbf{u}^\top \mathbf{Q}\mathbf{u})$ with
$\mathbf{Q} = c\,\mathrm{diag}(l_1 - l_1, l_1 - l_2, \dots, l_1 - l_p)$,
which is positive semidefinite, independent draws are obtained by the
rejection sampler of \citet{kent2018new} with an angular central Gaussian
envelope:

\bigskip

\begin{algorithm}[H]
\caption{Exact posterior sampler for $\boldsymbol{\Gamma}_1$ in the single-spike model}
\KwIn{$\mathbf{l}$, $n$, $\lambda_1$, $\sigma^2$, number of draws $M$}
Set $a_k = c\,(l_1 - l_k)$, $k = 1, \dots, p$\;
Solve $\sum_{k=1}^p (b + 2a_k)^{-1} = 1$ for $b \in (0, p]$; set $\kappa_k = 1 + 2a_k/b$\;
\Repeat{$M$ draws are accepted}{
  Draw $y_k \sim N(0, \kappa_k^{-1})$ independently and set $\mathbf{u} = \mathbf{y}/\|\mathbf{y}\|$\;
  Accept $\mathbf{u}$ with probability
  $\exp\!\big(-\sum_k a_k u_k^2\big)\big(\sum_k \kappa_k u_k^2\big)^{p/2}
   \exp\!\big(\tfrac{p-b}{2}\big)\big(\tfrac{b}{p}\big)^{p/2}$\;
}
\KwOut{$\boldsymbol{\Gamma}_1 = \mathbf{V}\mathbf{u}$ for each accepted $\mathbf{u}$}
\end{algorithm}

\section{Credible regions for the leading eigenvector}
\label{app:cr}
In this appendix, we use the output of the Gibbs samplers to construct confidence regions for the
leading eigenvector $\vec{\Gamma}_1$. We also use these regions to check that the oracle Gibbs sampler
produces draws from $\pi^{\opHaar} ( \vec{\Gamma} \mid \vec{x}, \vec{\lambda})$.

The confidence regions are $(1 - \alpha)$-level credible regions centered at $\vec{V}_1$,
the first column of the eigenvector MLE, $\vec{V}$,
which is a mode of $\pi^{\opHaar} ( \vec{\Gamma} \mid \vec{x}, \vec{\lambda})$.
We quantify the similarity of two vectors $\vec{W}_1, \vec{W}'_1 \in \mathbb{S}^{p-1}$ by their
absolute inner product $|\vec{W}_1^\top \vec{W}'_1|$, which does not depend on their signs.
Let $q^{\opHaar}_\alpha = q^{\opHaar}_\alpha(\vec{x};\vec{\lambda})$ denote the $\alpha$-quantile of $|\vec{V}_1^\top \vec{\Gamma}'_1|$, where
$\vec{\Gamma}'_1$ is the first column of $\vec{\Gamma}'$ sampled from
$\pi^{\opHaar} ( \vec{\Gamma} \mid \vec{x}, \vec{\lambda})$.
We then define the credible region for $\vec{\Gamma}_1$ as
\begin{equation} \label{def-CI}
CR ( \vec{V}, q^{\opHaar}_\alpha) = \{ \vec{W}_1 \in \mathbb{S}^{p-1} : \; |\vec{V}_1^\top \vec{W}_1| \ge q^{\opHaar}_\alpha \}.
\end{equation}
By Corollary~\ref{cor:cr} below, these credible regions also have exact frequentist coverage.

\begin{corollary}\label{cor:cr}
Suppose $\vec{\Sigma}=\vec{\Gamma}\vec{\Lambda}\vec{\Gamma}^\top$
with $\lambda_1>\lambda_2$.
If $\vec{X}$ has independent rows distributed as
$N_p(\vec{0},\vec{\Sigma})$, then
\[
  \Pr\!\left\{
    \vec{\Gamma}_1 \notin
    CR\!\left(\vec{V}(\vec{X}),
      q^{\opHaar}_\alpha(\vec{X};\vec{\lambda})\right)
  \right\}
  =\alpha,
\]
where $\vec{V}(\vec{X})$ is the sample eigenvector matrix.
\end{corollary}

\begin{proof}[Proof of Corollary~\ref{cor:cr}]
Write $\tilde{\vec{\Gamma}} = \vec{V}^\top\vec{\Gamma}$. Since $\vec{S} = \vec{V}\vec{L}\vec{V}^\top$,
(\ref{eq34a}) gives
\begin{align*}
    \pi^{\opHaar}(\vec{\Gamma} \mid \vec{x}, \vec{\lambda}) \propto
\exp[\optr\{-\vec{\Lambda}^{-1}\tilde{\vec{\Gamma}}^\top\vec{L}\tilde{\vec{\Gamma}}\,(n/2)\}],
\end{align*}
and as Haar measure is invariant under $\vec{\Gamma} \mapsto \vec{V}^\top\vec{\Gamma}$, the posterior
of $\tilde{\vec{\Gamma}}$ depends on $\vec{x}$ only through the sample eigenvalues $\vec{L}$.
Moreover, $|\vec{V}_1^\top \vec{\Gamma}_1'| = |\tilde\Gamma_{11}'|$. Hence
$q^{\opHaar}_\alpha = q^{\opHaar}_\alpha(\vec{x};\vec{\lambda})$ is a function of $\vec{L}$ only, and is unchanged when the
data are transformed by $\vec{x} \mapsto \vec{x}\vec{R}^\top$, $\vec{R} \in \opO(p)$. This
transformation maps $\vec{\Gamma}$ to $\vec{R}\vec{\Gamma}$ and $\vec{V}$ to
$\vec{R}\vec{V}$ (up to the signs of its columns), and
$|(\vec{R}\vec{W}_1)^\top \vec{R}\vec{W}_1'| = |\vec{W}_1^\top \vec{W}_1'|$.
The loss
$\ell(\vec{\Gamma}, \vec{x}) = I\{|\vec{V}_1^\top \vec{\Gamma}_1| < q^{\opHaar}_\alpha(\vec{x};\vec{\lambda})\}$
is therefore invariant under $\opO(p)$, and by Theorem~\ref{cor1} its risk
$\Pr(\vec{\Gamma}_1 \notin CR(\vec{V}, q^{\opHaar}_\alpha))$ is the same for
all $\vec{\Sigma}$ with eigenvalues $\vec{\lambda}$. It thus equals its average over
$\vec{\Gamma} \sim \pi^{\opHaar}$, which is the posterior probability of the event averaged
over the marginal distribution of $\vec{x}$. Since the posterior of $\vec{\Gamma}$ has a density with
respect to Haar measure, the posterior distribution of $|\vec{V}_1^\top \vec{\Gamma}_1'|$ is continuous,
so this posterior probability equals $\alpha$ for every $\vec{x}$, which proves the claim.
(The condition $\lambda_1 > \lambda_2$ ensures that $\vec{\Gamma}_1$ is determined up to sign, so that
the event $\vec{\Gamma}_1 \notin CR$ is well defined.)
\end{proof}

In the simulation study, we generated datasets with $p = 10$ and eigenvalues $\vec{\lambda} = (15, 5, \ldots, 5)$,
for $n=50$ and for $n=500$.
We constructed $0.90$ credible regions $CR(\vec{V}, q^{\opHaar}_{0.1})$, where $q^{\opHaar}_{0.1}$ is taken as the
$0.1$-quantile of $|\vec{V}_1^\top \vec{\Gamma}^{(g)}_1|$
based on posterior samples $\vec{\Gamma}^{(g)}$ drawn from the oracle eigenvalue Gibbs sampler.
Similarly, we constructed $CR(\vec{V}, q^{\ophBayes}_{0.1})$, where $q^{\ophBayes}_{0.1}$ is defined as the
$0.1$-quantile of $|\vec{V}_1^\top \vec{\Gamma}^{(g)}_1|$ using posterior samples $\vec{\Gamma}^{(g)}$
obtained from the hBayes Gibbs sampler.

Table~\ref{tab:oracle_hbayes_quantiles} summarizes the distribution of $|\vec{V}_1^\top \vec{\Gamma}^{(g)}_1|$ for
an $n = 50$ dataset and an $n = 500$ dataset, for the oracle and hBayes Gibbs samplers, and for exact samples
of the oracle posterior distribution, which are available only for the single-spike model.
The exact sampler, based on the rejection sampling algorithm of \citet{kent2018new}, is described in Appendix~\ref{app:spike}.
The nearly identical table summaries for the exact and oracle samplers suggest that the oracle sampler indeed samples
the oracle posterior distribution.
The table summaries also suggest that the distribution of $|\vec{V}_1^\top \vec{\Gamma}^{(g)}_1|$ for the hBayes
sampler is stochastically smaller than for the oracle sampler, suggesting that $CR(\vec{V}, q^{\ophBayes}_{0.1})$
achieves higher coverage probability than $CR(\vec{V}, q^{\opHaar}_{0.1})$.
For $n = 500$, the distributions of $|\vec{V}_1^\top \vec{\Gamma}^{(g)}_1|$ are nearly identical for
the two samplers and considerably stochastically larger than for $n = 50$.

For 200 simulated datasets with $n = 50$, the mean of $q^{\opHaar}_{0.1}$ was 0.8767 and the
coverage proportion of $CR(\vec{V}, q^{\opHaar}_{0.1})$ was 0.90; the mean of
$q^{\ophBayes}_{0.1}$ was 0.6658 and the coverage proportion of
$CR(\vec{V}, q^{\ophBayes}_{0.1})$ was 0.97. For $n = 500$, the corresponding values
were 0.9889 and 0.905 for the oracle regions, and 0.9886 and 0.915 for the hBayes regions. With 200
datasets, the Monte Carlo standard error of a coverage proportion near 0.9 is about 0.021, so the
coverage of the oracle regions is consistent with the nominal level for both sample sizes, as
Corollary~\ref{cor:cr} predicts.

The hBayes regions have no coverage guarantee, and for $n = 50$ they are conservative. In some
iterations of the hBayes sampler the sampled leading eigenvalue falls among the bulk eigenvalues, and the
corresponding draws of $\vec{\Gamma}_1$ are far from $\vec{V}_1$. This explains the small minimum of
$|\vec{V}_1^\top \vec{\Gamma}^{(g)}_1|$ for the hBayes sampler in
Table~\ref{tab:oracle_hbayes_quantiles}. For $n = 500$ the hBayes regions have close to nominal coverage.

\begin{table}[ht]
\centering
\begin{tabular}{lccccccc}
\toprule
  & & Min. & 1st Qu. & Median & Mean & 3rd Qu. & Max. \\
\midrule
$n=50$ & Exact sampler  &   0.7787  & 0.9344 & 0.9523 & 0.9479 & 0.9663 & 0.9978 \\
        & Oracle sampler  &  0.8440 & 0.9338 & 0.9527 & 0.9483 & 0.9669 & 0.9949  \\
        & hBayes sampler & 0.0039 & 0.9236 & 0.9480 & 0.9266 & 0.9648 & 0.9965 \\ \hline
$n=500$ & Exact sampler  &   0.9732 & 0.9916 & 0.9939 & 0.9934 & 0.9957 & 0.9993   \\
        & Oracle sampler  &  0.9724 &  0.9918 & 0.9940 & 0.9935 & 0.9958 & 0.9994  \\
        & hBayes sampler & 0.9773 & 0.9915 & 0.9938 & 0.9933 & 0.9957 & 0.9996  \\
\bottomrule
\end{tabular}
\caption{Summary of the posterior distribution of $|\vec{V}_1^\top \vec{\Gamma}^{(g)}_1|$ generated by the exact sampler, oracle sampler
and hBayes sampler, for one simulated dataset with $n = 50$ and one with $n = 500$; $p = 10$ and
$\vec{\lambda} = (15, 5, \ldots, 5)$.}
\label{tab:oracle_hbayes_quantiles}
\end{table}

\section{Full results of the risk comparison}
\label{app:risk-full}
Table~\ref{tab:risk-full} extends Table~\ref{tab:risk} of Section~\ref{sec:risk} to the clustered
configuration and adds the sample covariance matrix (MLE) and the estimator of \citet{JamesStein1961} (JS).

\begin{table}[!ht]
\centering
\linespread{1}\scriptsize
\setlength{\tabcolsep}{4pt}
\begin{tabular}{llrrrrrr}
\hline
 & & \multicolumn{3}{c}{$\boldsymbol\Sigma$} & \multicolumn{3}{c}{$\boldsymbol\Omega$} \\
\cmidrule(lr){3-5}\cmidrule(lr){6-8}
 & & $L_0$ & $L_1$ & $L_2$ & $L_0$ & $L_1$ & $L_2$ \\
\hline\hline
\multicolumn{8}{l}{\emph{Linear}} \\
\hline
$n=30$ & oracle & 4174 (15) & 13.76 (0.05) & 17.04 (0.09) & 0.999 (0.002) & 9.10 (0.02) & 9.16 (0.01) \\
 & eBayes & \textbf{4390 (29)} & \textbf{14.67 (0.10)} & 20.75 (0.50) & \textbf{1.078 (0.015)} & \textbf{9.33 (0.04)} & \textbf{9.46 (0.03)} \\
 & hBayes & 4537 (37) & 14.80 (0.13) & \textbf{17.57 (0.11)} & 1.191 (0.038) & 9.39 (0.04) & 9.68 (0.03) \\
 & LW.lin & 4475 (21) & 28.37 (0.22) & 341.04 (4.89) & 1.243 (0.002) & 9.40 (0.02) & 12.89 (0.16) \\
 & LW.nl & 4480 (32) & 26.63 (0.34) & 299.60 (7.47) & 1.223 (0.004) & 9.43 (0.04) & 13.87 (0.30) \\
 & glasso & -- & -- & -- & 1.259 (0.002) & 11.58 (0.04) & 19.04 (0.31) \\
\hline
$n=50$ & oracle & 3492 (15) & 9.50 (0.06) & 12.87 (0.07) & 0.701 (0.012) & 7.77 (0.03) & 8.76 (0.02) \\
 & eBayes & \textbf{3587 (19)} & 10.32 (0.10) & 15.69 (0.54) & 0.835 (0.009) & \textbf{7.95 (0.04)} & \textbf{8.91 (0.03)} \\
 & hBayes & 3656 (23) & \textbf{10.01 (0.08)} & \textbf{13.43 (0.22)} & \textbf{0.808 (0.011)} & 7.98 (0.04) & 9.01 (0.03) \\
 & MLE & 14035 (123) & 24.77 (0.14) & 33.05 (0.24) & 94.054 (6.448) & 154.48 (3.57) & 6625.26 (436.82) \\
 & JS & 10564 (98) & 19.25 (0.11) & 30.12 (0.34) & 11.616 (0.638) & 71.22 (1.29) & 1251.00 (57.91) \\
 & Haff & 4666 (32) & 10.34 (0.09) & 24.28 (0.84) & 1.536 (0.143) & 19.62 (0.59) & 156.48 (11.82) \\
 & Stein & 4335 (42) & 11.77 (0.19) & 42.67 (1.87) & 1.506 (0.143) & 15.91 (0.56) & 117.46 (11.46) \\
 & LW.lin & 3925 (16) & 24.95 (0.16) & 275.90 (3.27) & 1.216 (0.002) & 8.65 (0.02) & 11.85 (0.11) \\
 & LW.nl & 3718 (24) & 17.18 (0.26) & 125.51 (4.63) & 1.134 (0.066) & 8.75 (0.24) & 22.57 (4.37) \\
 & glasso & -- & -- & -- & 1.240 (0.002) & 10.52 (0.04) & 16.35 (0.22) \\
\hline
$n=100$ & oracle & 2513 (12) & 5.23 (0.03) & 8.40 (0.05) & 0.244 (0.007) & 4.94 (0.03) & 7.18 (0.03) \\
 & eBayes & \textbf{2542 (14)} & 5.33 (0.03) & 8.37 (0.04) & 0.318 (0.006) & \textbf{5.05 (0.03)} & \textbf{7.30 (0.03)} \\
 & hBayes & 2564 (15) & \textbf{5.31 (0.03)} & \textbf{8.31 (0.04)} & \textbf{0.305 (0.006)} & \textbf{5.05 (0.03)} & 7.33 (0.03) \\
 & MLE & 6919 (56) & 9.61 (0.05) & 16.39 (0.10) & 1.800 (0.078) & 18.23 (0.17) & 101.42 (1.76) \\
 & JS & 5812 (42) & 8.48 (0.04) & 14.96 (0.10) & 0.726 (0.027) & 14.69 (0.14) & 71.27 (1.16) \\
 & Haff & 3036 (19) & 5.41 (0.03) & 11.24 (0.12) & 0.312 (0.011) & 7.16 (0.07) & 24.88 (0.45) \\
 & Stein & 2928 (22) & 5.92 (0.07) & 15.34 (0.49) & 0.321 (0.011) & 6.54 (0.08) & 20.20 (0.54) \\
 & LW.lin & 3034 (12) & 18.89 (0.10) & 175.92 (1.70) & 1.154 (0.001) & 7.18 (0.02) & 9.83 (0.07) \\
 & LW.nl & 2572 (15) & 6.95 (0.05) & 24.32 (0.32) & 0.435 (0.011) & 5.13 (0.03) & 10.34 (0.14) \\
 & glasso & -- & -- & -- & 1.142 (0.003) & 8.46 (0.03) & 14.12 (0.14) \\
\hline\hline
\multicolumn{8}{l}{\emph{Clustered}} \\
\hline
$n=30$ & oracle & 5104 (16) & 14.48 (0.05) & 17.23 (0.09) & 0.961 (0.001) & 8.96 (0.02) & 8.48 (0.01) \\
 & eBayes & \textbf{5423 (38)} & 15.74 (0.14) & 23.79 (1.11) & \textbf{1.058 (0.019)} & \textbf{9.20 (0.04)} & \textbf{8.85 (0.03)} \\
 & hBayes & 5620 (50) & \textbf{15.60 (0.15)} & \textbf{18.01 (0.24)} & 1.141 (0.038) & 9.27 (0.04) & 9.08 (0.04) \\
 & LW.lin & 5542 (27) & 31.53 (0.24) & 453.25 (6.41) & 1.180 (0.001) & 9.29 (0.02) & 11.81 (0.15) \\
 & LW.nl & 5522 (43) & 29.68 (0.42) & 402.01 (10.97) & 1.164 (0.003) & 9.31 (0.04) & 12.73 (0.30) \\
 & glasso & -- & -- & -- & 1.189 (0.002) & 11.36 (0.04) & 17.59 (0.31) \\
\hline
$n=50$ & oracle & 4223 (18) & 9.52 (0.06) & 12.75 (0.07) & 0.624 (0.014) & 7.59 (0.04) & 8.16 (0.02) \\
 & eBayes & \textbf{4370 (22)} & 10.60 (0.12) & 16.16 (0.67) & 0.795 (0.009) & \textbf{7.83 (0.03)} & \textbf{8.35 (0.02)} \\
 & hBayes & 4463 (26) & \textbf{10.21 (0.10)} & \textbf{13.59 (0.31)} & \textbf{0.760 (0.010)} & 7.86 (0.04) & 8.46 (0.02) \\
 & MLE & 19068 (163) & 24.77 (0.14) & 33.05 (0.24) & 76.888 (5.365) & 154.48 (3.57) & 6625.26 (436.82) \\
 & JS & 14418 (134) & 19.25 (0.11) & 30.13 (0.34) & 8.976 (0.514) & 71.22 (1.29) & 1251.27 (57.99) \\
 & Haff & 5823 (43) & 10.34 (0.10) & 24.64 (0.85) & 1.294 (0.117) & 19.82 (0.60) & 162.94 (12.13) \\
 & Stein & 5326 (56) & 11.82 (0.18) & 44.86 (2.09) & 1.253 (0.117) & 15.55 (0.55) & 112.34 (10.85) \\
 & LW.lin & 4906 (18) & 28.06 (0.18) & 372.98 (4.42) & 1.160 (0.001) & 8.63 (0.01) & 11.03 (0.10) \\
 & LW.nl & 4523 (29) & 17.82 (0.31) & 142.93 (6.32) & 1.063 (0.054) & 8.70 (0.26) & 22.54 (4.28) \\
 & glasso & -- & -- & -- & 1.178 (0.002) & 10.41 (0.04) & 15.05 (0.21) \\
\hline
$n=100$ & oracle & 2985 (14) & 5.08 (0.03) & 8.21 (0.05) & 0.194 (0.006) & 4.73 (0.03) & 6.73 (0.03) \\
 & eBayes & \textbf{3035 (17)} & 5.22 (0.03) & 8.26 (0.04) & 0.272 (0.006) & \textbf{4.86 (0.03)} & \textbf{6.86 (0.03)} \\
 & hBayes & 3062 (18) & \textbf{5.20 (0.03)} & \textbf{8.21 (0.04)} & 0.260 (0.006) & \textbf{4.86 (0.03)} & 6.89 (0.03) \\
 & MLE & 9404 (73) & 9.61 (0.05) & 16.39 (0.10) & 1.510 (0.074) & 18.23 (0.17) & 101.42 (1.76) \\
 & JS & 7917 (57) & 8.48 (0.04) & 14.96 (0.10) & 0.565 (0.024) & 14.69 (0.14) & 71.28 (1.16) \\
 & Haff & 3699 (25) & 5.28 (0.03) & 11.09 (0.12) & \textbf{0.254 (0.010)} & 6.93 (0.07) & 23.99 (0.45) \\
 & Stein & 3571 (31) & 5.74 (0.06) & 14.79 (0.40) & 0.260 (0.011) & 6.34 (0.08) & 19.63 (0.54) \\
 & LW.lin & 3859 (14) & 21.78 (0.12) & 245.72 (2.29) & 1.115 (0.001) & 7.33 (0.02) & 9.43 (0.06) \\
 & LW.nl & 3085 (18) & 6.81 (0.05) & 24.13 (0.33) & 0.381 (0.011) & 4.94 (0.03) & 9.76 (0.14) \\
 & glasso & -- & -- & -- & 1.109 (0.003) & 8.61 (0.03) & 13.36 (0.13) \\
\hline
\end{tabular}
\caption{Mean loss (standard error) over 100 datasets, $p=40$, both eigenvalue configurations. The oracle rule uses the true eigenvalues. Bold: smallest mean loss among the non-oracle estimators. For $n=30<p$ the MLE is singular and the JS, Haff and Stein estimators are not defined, so they are omitted; the graphical lasso estimates only $\boldsymbol\Omega$.}
\label{tab:risk-full}
\end{table}

\end{document}